\documentclass[twoside]{article}
\usepackage{epsfig}
\input{qic.sty}

\let\square\relax
\usepackage{amsmath,amsfonts}

\let\captionbox\relax
\usepackage{caption}
\usepackage{pgfplots}
\usepackage{braket}
\usepackage{thmtools}
\usepackage{thm-restate}
\usepackage{appendix}

\usepackage{algorithm}
\usepackage[noend]{algpseudocode}

\usepackage[numbers]{natbib} 
\usepackage{tikz}
\usetikzlibrary{positioning}
\usepackage{pgfplots}
\pgfplotsset{compat=1.14}

\newcommand\xqed[1]{%
  \leavevmode\unskip\penalty9999 \hbox{}\nobreak\hfill
  \quad\hbox{#1}}

\newenvironment{proof}{\textit{Proof.\;\;}}{\xqed{$\square$}\\}

\newcommand{\integer}{\mathbb{Z}}
\renewcommand{\natural}{\mathbb{N}}

\newcommand{\Hi}{\mathcal{H}}

\renewcommand{\P}{P}
\newcommand{\U}{U}
\newcommand{\D}{D}
\newcommand{\I}{I}

\newcommand{\R}{R}
\newcommand{\V}{\mathcal{V}}
\newcommand{\E}{\mathcal{E}}
\renewcommand{\S}{\mathcal{S}}
\newcommand{\T}{\mathcal{T}}

\renewcommand{\Pr}{{\mathit{\Pi}}}

\newcommand{\0}{\flat}

\def\h #1{\hat{ #1 }}

\newcommand{\amax}{\overline{a}}
\newcommand{\amin}{\underline{a}}
\newcommand{\phin}{\Phi_{\mathrm{in}}}
\newcommand{\phout}{\Phi_{\mathrm{out}}}

\begin{document}
\setlength{\textheight}{8.0truein}

\runninghead{Quantum Fast-Forwarding: Markov Chains and Graph Property Testing}
            {Simon Apers and Alain Sarlette}

\normalsize\textlineskip
\thispagestyle{empty}
\setcounter{page}{1}

\copyrightheading{0}{0}{2018}{000--000}

\vspace*{0.88truein}

\fpage{1}

\centerline{\bf
QUANTUM FAST-FORWARDING:}
\vspace*{.1truein}
\centerline{\bf
MARKOV CHAINS AND GRAPH PROPERTY TESTING}
\vspace*{0.37truein}
\centerline{\footnotesize
SIMON APERS}
\vspace*{0.015truein}
\centerline{\footnotesize\it
Team SECRET, INRIA Paris, France}
\vspace*{10pt}
\centerline{\footnotesize 
ALAIN SARLETTE}
\vspace*{0.015truein}
\centerline{\footnotesize\it 
QUANTIC lab, INRIA Paris, France}
\baselineskip=10pt
\centerline{\footnotesize\it
Department of Electronics and Information Systems, Ghent University, Belgium}
\baselineskip=10pt
\vspace*{.1truein}
\centerline{\footnotesize\it \{simon.apers@inria.fr,alain.sarlette@inria.fr\}}

\vspace*{0.21truein}

\abstracts{
We introduce a new tool for quantum algorithms called quantum fast-forwarding (QFF).
The tool uses quantum walks as a means to quadratically fast-forward a reversible Markov chain.
More specifically, with $P$ the Markov chain transition matrix and $D = \sqrt{P\circ P^T}$ its discriminant matrix ($D=P$ if $P$ is symmetric), we construct a quantum walk algorithm that for any quantum state $\ket{v}$ and integer $t$ returns a quantum state $\epsilon$-close to the state $D^t\ket{v}/\|D^t\ket{v}\|$.
The algorithm uses $O\Big(\|D^t\ket{v}\|^{-1}\sqrt{t\log(\epsilon\|D^t\ket{v}\|)^{-1}}\Big)$ expected quantum walk steps and $O(\|D^t\ket{v}\|^{-1})$ expected reflections around $\ket{v}$.
This shows that quantum walks can accelerate the transient dynamics of Markov chains, complementing the line of results that proves the acceleration of their limit behavior.
}{We show that this tool leads to speedups on random walk algorithms in a very natural way.
Specifically we consider random walk algorithms for testing the graph expansion and clusterability, and show that we can quadratically improve the dependency of the classical property testers on the random walk runtime.
Moreover, our quantum algorithm exponentially improves the space complexity of the classical tester to logarithmic.
As a subroutine of independent interest, we use QFF for determining whether a given pair of nodes lies in the same cluster or in separate clusters.
This solves a robust version of $s$-$t$ connectivity, relevant in a learning context for classifying objects among a set of examples.
The different algorithms crucially rely on the quantum speedup of the transient behavior of random walks.
}{}

\vspace*{10pt}

\vspace*{3pt}

\vspace*{1pt}\textlineskip

\section{Introduction and Summary}
Quantum walks (QWs) have been shown to provide a speedup over classical Markov chains in a variety of settings.
In the class of search problems, there exist quantum walk algorithms that accelerate tasks such as detecting element distinctness \cite{ambainis2007quantum}, finding triangles \cite{magniez2007quantum}, and hitting marked elements \cite{childs2003exponential,szegedy2004quantum,krovi2016quantum}.
In the class of sampling problems, there exist quantum walk algorithms that speed up mixing on graphs \cite{ambainis2001one,aharonov2001quantum,richter2007quantum} and simulated annealing \cite{somma2008quantum,wocjan2008speedup}, and allow for quantum state generation \cite{aharonov2003adiabatic}.
A broader overview is given in the surveys by Ambainis \cite{ambainis2003quantum} and Santha \cite{santha2008quantum}.

In this work we further develop this list by showing that quantum walks can be used in a very natural way to speed up random walk algorithms for graph property testing.
Central to this result is a new tool which we call quantum walk fast-forwarding, allowing to quadratically fast-forward the full dynamics of a reversible Markov chain.
Whereas most existing quantum walk algorithms build on a quadratic speedup towards the Markov chain limit behavior, quantum fast-forwarding allows to accelerate the transient dynamics as well.
This feature is crucial towards speeding up the classical algorithms for property testing.

\subsection*{Quantum Walk Fast-Forwarding}
Many of the above mentioned algorithms are to some extent preceded and inspired by the work of Watrous \cite{watrous2001quantum}.
In this work, he introduced quantum walks as a means to \textit{quantum simulate} random walks as a superposition on a quantum computer, without resorting to intermediate measurements.
With $P$ the transition matrix of a random walk on a regular graph, and $\ket{v}$ some arbitrary initial quantum state, he shows that it is possible to create the quantum state
\[
\ket{P^t v}
= P^t\ket{v}/\|P^t\ket{v}\|
\]
using $O(\|P^t\ket{v}\|^{-2}\, t)$ expected QW steps, and $O(\|P^t\ket{v}\|^{-2})$ expected copies of $\ket{v}$.
This allowed him to quantum simulate the famous random walk algorithm by Aleliunas et al \cite{aleliunas1979random} for undirected graph connectivity, thereby proving that the complexity class symmetric logspace is contained in a quantum analogue of randomized logspace.

In this work we show that quantum walks can create the state $\ket{P^t v}$ quadratically faster.
Indeed, we show that \textit{quantum walks can quadratically fast-forward a general reversible Markov chain}.
More specifically, let $P$ be the transition matrix of a reversible Markov chain on a finite state space with discriminant matrix $D = \sqrt{P\circ P^T}$, where the square root and ``$\circ$''-product are elementwise.
Note that if $P$ is symmetric, as in the work of Watrous, then $D = P$.
Following the work of Szegedy \cite{szegedy2004quantum}, generalizing the approach in \cite{watrous2001quantum}, we can associate a quantum walk to $P$ whose spectral properties are closely tied to those of $P$.
These results provide the ground for most existing quantum walk algorithms, building on a quadratic speedup of the Markov chain limit behavior.
For intermediate times however the behavior of these quantum walks will in general be unrelated to the Markov chain behavior.
We prove that applying a technique called \textit{linear combination of unitaries} \cite{childs2012hamiltonian,berry2015simulating,van2017quantum} on the QW operator allows to mediate this shortcoming.
Indeed, combining this technique with a truncated Chebyshev expansion of the Markov chain eigenvalue function allows to simulate and accelerate the (spectral) dynamics of the Markov chain.
We name this scheme \textit{quantum walk fast-forwarding} (QFF), and it condenses into the following theorem:
\begin{theorem}[Quantum walk fast-forwarding with reflection] \label{thm:QFFg-intro}
Given any quantum state $\ket{v}$, $t\geq 0$ and $\epsilon>0$, QFFg (Algorithm \ref{alg:QFFg}) outputs a quantum state $\epsilon$-close to $\ket{D^t v}$ using
\[
O\Big( \|D^t\ket{v}\|^{-1} \sqrt{t\log(\epsilon \|D^t\ket{v}\|)^{-1}} \Big)
\]
expected QW steps and $O(\|D^t\ket{v}\|^{-1})$ expected reflections around $\ket{v}$.
\end{theorem}

Much of the previous work that builds on Szegedy's quantum walk, such as \cite{szegedy2004quantum,somma2008quantum}, relies on the quadratic improvement of the spectral gap when compared to the original Markov chain.
This suffices when one is interested in the limit behavior of the dynamics.
Our result, however, captures the transient dynamics which are governed by the complete spectrum of eigenvalues and corresponding eigenvectors.
Similarly to both the preceding work and the existing classical algorithms, our algorithm makes use of only local information on the graph and Markov chain.
Indeed we show that our algorithm allows quantum walks to simulate the dynamics of this entire classical spectrum, all the while retaining a quadratic acceleration\footnote{This is reminiscent of the work by Miclo and Diaconis \cite{diaconis2013spectral} on second order Markov chains, where they show that decreasing the probability that a Markov chain backtracks improves not only the spectral gap, but the entire spectrum.
In contrast to quantum walks, however, this improvement will generally only be a constant factor, rather than quadratic.}.

Upon completion of this work, we became aware of the recent work on quantum singular value transformation by Gily\'en, Su, Low and Wiebe \cite{gilyen2018quantum}.
This work generalizes a wide range of advances in quantum algorithms for Hamiltonian simulation, Gibbs sampling and others.
In Section \ref{subsec:QSVT} we discuss how our algorithm and its properties can alternatively be proved using this framework.

\subsection*{Quantum Graph Property Testing}

We will show that QFF allows to very naturally speed up random walk algorithms for graph property testing.
Given query access to a graph, property testing aims to determine whether it has a certain property, or whether it is far from having this property.
Among the graphs with degree bound $d$, as we will be focusing on, two $N$-node graphs are said to be $\epsilon$-far from each other if at least $\epsilon d N$ edges have to be removed or added to turn one graph into the other.
As an example, one can ask whether a given graph is bipartite, or whether it is at least $\epsilon$-far from any bipartite graph.
Testing bipartiteness is a relaxation as compared to effectively deciding whether the graph is bipartite or not (but possibly very close to bipartite), allowing for algorithms to work in sublinear time, i.e., scale as $o(N)$ with $N$ the number of nodes in the graph.
This is in contrast to the complexity of deciding properties exactly, which typically requires a number of queries at least linear in the graph size.
In many realistic settings, see for instance the discussion of \textit{massive graphs} in \cite{spielman2013local}, linear in the graph size is no longer computationally feasible, hence sparking the interest in sublinear time algorithms.

We will consider property testers for the expansion and the clusterability of graphs.
We start by discussing the expansion tester of Goldreich and Ron (GR) \cite{goldreich2011testing}, and we prove how QFF allows to accelerate this tester.
Specifically the problem is to determine whether the given graph has vertex expansion $\geq \Upsilon$, or whether it is $\epsilon$-far from any such graph.
The expansion of a graph forms a measure for the random walk mixing time over the graph.
The idea behind the GR tester is therefore to run a number of random walks and count the number of pairwise collisions between the end points.
If a random walk is congested in some low expansion set, then this number will be greater than when the random walk mixes efficiently.
It thus forms a measure for the mixing behavior and expansion of the random walk.
The runtime of their algorithm is
\[
O(N^{1/2+\mu} \Upsilon^{-2} d^2 \epsilon^{-1} \log N),
\]
with the $d^2\Upsilon^{-2}$-factor determined by the random walk runtime.

We show that QFF very naturally allows to speed up this algorithm by fast-forwarding the random walk, and then using quantum amplitude estimation to estimate the 2-norm of the random walk probability distribution.
This 2-norm will similarly be large if the random walk congests and small otherwise, thus allowing to detect whether the random walk is able to efficiently spread out or not.
The runtime of our quantum algorithm is
\[
O(N^{1/2+\mu} \Upsilon^{-1} d^{3/2} \epsilon^{-1} \log N),
\]
which basically follows from quadratically improving the random walk runtime.
In addition, our algorithm only requires $\mathrm{polylog}(N)$ space, as compared to the $\mathrm{poly}(N)$ space requirements of the GR tester.
We note that in preceding work Ambainis, Childs and Liu \cite{ambainis2011quantum} have also used quantum walks to speed up the GR tester, be it in an indirect way.
Roughly they apply Ambainis' element distinctness algorithm \cite{ambainis2007quantum} to speed up the search of collisions between random walk end points from $N^{1/2}$ to $N^{1/3}$.
Compared to our result, they find a complimentary speedup to $O(N^{1/3+\mu} \Upsilon^{-2} d^2 \epsilon^{-1} \log N)$.
Due to the use of the element distinctness algorithm, their algorithm does require $\mathrm{poly}(N)$ space.

We continue by discussing the more recent line of algorithms for testing graph clusterability \cite{czumaj2015testing,chiplunkar2018testing}, forming a natural generalization of the work of Goldreich and Ron.
We discuss how these techniques make use of algorithms for classifying nodes in clusters, and show how QFF allows to accelerate these algorithms.
Such node classification is of relevance beyond the setting of property testing, allowing for instance nearest-neighbor classification of nodes in a learning problem.

We remark that work by Valiant and Valiant \cite{valiant2011testing} shows that estimating the distance in 2-norm between given probability distributions is much easier and more stable than estimating the distance in 1-norm, which would otherwise be the natural choice.
This underlies the fact that many graph property testing algorithms estimate the 2-distance between random walk distributions.
QFF allows to cast a probability distribution $p$ as a quantum state $\ket{p} = p/\|p\|$, which is naturally associated to the 2-norm.
As a consequence, QFF very naturally leads to quantum algorithms for estimating the 2-norm distance between random walk distributions, directly leading to the quantization and speedup of the above graph property testers.

\section{Quantum Walk Fast-Forwarding} \label{sec:quantumFF}

In this section we elaborate the details of the quantum walk fast-forwarding scheme.
First, we formally introduce the concept and characteristics of a quantum walk associated to a reversible Markov chain.
These results provide the ground for most existing quantum walk algorithms, building on a quadratic speedup of the Markov chain limit behavior.
We discuss how these results fall short for speeding up any transient behavior of the Markov chain.
Second, we prove how a technique called \textit{linear combinations of unitaries} can be used to mediate this shortcoming.
By combining this technique with a truncated Chebyshev expansion of the general Markov chain eigenvalue function, we arrive at our quantum algorithm for quantum walk fast-forwarding.

\subsection{Preliminaries: Quantum Walk Schemes}
In this section we review the aforementioned quantum walk scheme by Watrous \cite{watrous2001quantum}, and show how it gives rise to the subsequent work on quantum walk speedups by Ambainis \cite{ambainis2007quantum}, Szegedy \cite{szegedy2004quantum}, Magniez et al \cite{magniez2011search}, and many others.
Apart from a new proof of Proposition \ref{prop:QW-szegedy}, the results in this section are known, and if necessary a reader could skip the section.
For the rest of this paper we will only consider simple graphs $G=(\V,\E)$ with node set $\V$ and edge set $\E \subseteq \V \times \V$.
We will also refer to a Markov chain by its stochastic transition matrix $\P$.

\subsubsection{Watrous Scheme}
Consider a Markov chain $P$ on a graph $G=(\V,\E)$, and an initial probability distribution $v$ over $\V$.
In early work, Watrous \cite{watrous2001quantum} proposed a quantum walk scheme for creating the quantum state $\ket{P^t v}$ associated to the classical distribution $P^t v$, defined by
\begin{equation} \label{eq:def-qs}
\Ket{P^t v}
= \frac{1}{\|P^t v\|} \sum (P^t v)(j) \, \ket{j},
\end{equation}
where $(P^t v)(j)$ denotes the $j$-th component of the probability vector $P^t v$.
For a general nonzero vector $w$, we will use the notation $\ket{w} = \frac{1}{\|w\|} \sum w(j) \ket{j}$ which associates a quantum state $\ket{w}$ to $w$.
The quantum walk associated to $P$ takes places on the extended or ``coined'' node space $\hat{\V} = \V \times \{\0,\V\} = \{ (i,j) \mid i\in\V, j\in\{\0,\V\} \}$, where ``$\0$'' denotes some canonical initialization state.
The associated Hilbert space is $\Hi = \mathrm{span} \{ \ket{i,j} \mid (i,j) \in \hat{\V}\}$.
We will call the subspace $\Hi_\0 = \mathrm{span} \{ \ket{i,\0} \mid i\in\V \}$ the \textit{flat subspace}, associated to the projector $\Pr_\0 = \I \otimes \ket{\0}\bra{\0}$, with $I$ the identity operator on the first register.
The discrete-time quantum walk is described by a unitary operator $\U_{\P}$ on $\Hi$, defined by a \textit{shift operator} $S$ and a \textit{coin toss operator} $V$,
\begin{equation} \label{eq:QW-watrous}
\U_{\P}
= V^\dag S V.
\end{equation}
We will often write $\U$ instead of $\U_{\P}$ when the context allows it.
The coin toss operator is defined as $V = \sum_i \ket{i}\bra{i} \otimes V_i$, where $V_i$ is such that
\begin{equation} \label{eq:watrous-coin-toss}
V \ket{i,\0}
= \ket{i} \otimes V_i \ket{\0}
= \ket{i} \otimes \ket{\psi_i}
= \ket{i} \otimes \sum_j \sqrt{\P(j,i)} \ket{j}.
\end{equation}
By the design of the QW scheme, as we will see later, it suffices to characterize the action of $V_i$ on the state $\ket{\0}$.
The operators $V_i$ can then be arbitrarily completed into unitary matrices.
The shift operator is defined by the permutation
\[
\ket{i,j}
\mapsto
S \ket{i,j} =
\begin{cases}
\ket{j,i},\; &(i,j) \in \E, \\
\ket{i,j},\; &\text{otherwise,}
\end{cases}
\]
and $S \ket{i,\0} = S\ket{i,\0}$.
It is now easy to prove the below lemma, stating that the restriction of $\U$ to the flat subspace implements the \textit{discriminant matrix}
\[
\D = \sqrt{\P \circ \P^T},
\]
with the square root and ``$\circ$''-product elementwise.
The discriminant matrix is closely related to the original Markov chain $\P$, and if $\P$ is reversible then they share the same eigenvalues.
We will often write $\D\ket{v,\0}$ as shorthand for $(\D\otimes\I)\ket{v,\0}$.

\begin{lemma}[\cite{watrous2001quantum}] \label{lem:QW-watrous}
For any quantum state $\ket{v,\0}$, it holds that
\[
\Pr_{\0} \U \ket{v,\0}
= \D\ket{v,\0}.
\]
\end{lemma}
\begin{proof}
This directly follows from the fact that for any node $i$ it holds that $\U \ket{i,\0} = \sum_j \sqrt{\P(i,j)\P(j,i)} \ket{j,\0} + \ket{\psi^\perp}$, where $\ket{\psi^\perp}$ is some state perpendicular to the flat subspace $\Hi_\0$.
By linearity, the proposition follows for general $\ket{v,\0}$.
\end{proof}

\noindent
This gives rise to an easy QW algorithm for creating $\ket{D^t v}$.
Namely, do $t$ times: (i) apply a single step of the QW $\U$, (ii) perform the measurement corresponding to the measurement operators $\Pr_\0$ and $\I-\Pr_\0$.
If each of the measurements returns ``$\0$'', which happens with a probability
\[
\frac{\|\D^t\ket{v}\|^2}{\|\D^{t-1}\ket{v}\|^2}\frac{\|\D^{t-1}\ket{v}\|^2}{\|\D^{t-2}\ket{v}\|^2}
	\dots \frac{\|\D\ket{v}\|^2}{\|\ket{v}\|^2}
= \|\D^t\ket{v}\|^2,
\]
then the output state is $\ket{\D^t v,\0}$.
For symmetric $\P$, as in Watrous' original paper, it holds that $\D = \P$, and so this approach effectively returns the quantum state $\ket{\P^t v}$ that we were looking for.
It requires $t$ QW steps and succeeds with a probability $\|\D^t\ket{v}\|^2$.

\subsubsection{Quadratically Improved Spectrum}
The main idea of our new QFF tool is that we can quadratically accelerate the number of QW steps required: we can create the state $\ket{\D^t v}$ using $O(\sqrt{t})$ QW steps, succeeding with the same probability $\|\D^t\ket{v}\|^2$.
To do so we make use of work that followed up on the QW approach by Watrous, mainly initiated by Ambainis \cite{ambainis2007quantum} and Szegedy \cite{szegedy2004quantum} with the aim of accelerating classical search problems.
We will discuss how in a certain sense this operator quadratically improves the Markov chain spectrum, yet falls short of speeding up its full dynamics.
In the next section we then present our more fine-grained scheme that resolves this issue.

They proposed an alternative QW, essentially adding a reflection around the flat subspace $\Hi_\0$ to the QW operator $\U$ by Watrous:
\begin{equation} \label{eq:QW-ambainis}
W
= \R_\0 \U
= (2\Pr_\0 - \I) \U.
\end{equation}
Their key insight is captured in the following proposition, for which we provide a new proof which explicitly builds on the insight from Watrous' work.
We will denote by $T_t$ the $t$-th Chebyshev polynomial of the first kind.

\begin{proposition}[\cite{szegedy2004quantum}] \label{prop:QW-szegedy}
For any $\ket{v,\0}$, it holds that
\[
\Pr_\0 W^t \ket{v,\0}
= T_t(\D) \ket{v,\0}.
\]
As a consequence, if $(\cos\theta,\ket{v})$ is an eigenpair of $\D$, then
\[
\Pr_\0 W^t \ket{v,\0}
= T_t(\cos\theta) \ket{v,\0}
= \cos(t\theta) \ket{v,\0}.
\]
\end{proposition}
\begin{proof}
We easily find a recursion formula for $\Pr_\0 W^t$:
\[
\Pr_\0 W^t
= \Pr_\0 \R_\0 \U (2\Pr_\0 - \I) \U W^{t-2}
= 2\Pr_\0 \U (\Pr_\0 W^{t-1}) - \Pr_\0 W^{t-2}
\]
using the fact that $\Pr_\0\R_\0 = \Pr_\0$, and $\U^\dag = \U$ so that $\U^2 = \U\U^\dag =\I$.
Since $\Pr_\0 W^0 = \Pr_\0$ and $\Pr_\0 W = \Pr_\0 \U$, this shows that we can express $\Pr_\0 W^t$ as a polynomial in $\Pr_\0 \U$.
The Chebyshev polynomials of the first kind $T_t$ are defined by
\[
T_0(x) = 1,\quad
T_1(x) = x,\quad
T_t(x) = 2xT_{t-1}(x) - T_{t-2}(x).
\]
Setting $x = \Pr_\0 \U$ and $T_0(\Pr_\0 \U) = \Pr_\0$, this shows that we can express $\Pr_\0 W^t$ as $\Pr_\0 W^t
= T_t(\Pr_\0 \U)$.
From Lemma \ref{lem:QW-watrous} we know that $(\Pr_\0 \U)^t \ket{v,\0} = \D^t \ket{v,\0}$, and therefore
\[
\Pr_\0 W^t \ket{v,\0}
= T_t(\D) \ket{v,\0}.
\]
Using the geometric definition of $T_t$, $T_t(\cos\theta) = \cos(t\theta)$, we see that if $\D \ket{v,\0} = \cos\theta \ket{v,\0}$ then
\[
\Pr_\0 W^t \ket{v,\0} = T_t(\cos\theta) \ket{v,\0} = \cos(t\theta) \ket{v,\0}.
\]
\end{proof}
This proposition constitutes the basis from which most of the aforementioned quantum algorithms start, and it will be the basis from which this work starts.
The gist of the speedup comes from comparing the original action of $\D^t$ on an eigenpair $(\cos\theta,\ket{v})$, $\D^t \ket{v,\0} = \cos^t(\theta) \ket{v,\0}$, with the action of $\Pr_\0 W^t$, $\Pr_\0 W^t \ket{v,\0} = \cos(t\theta) \ket{v,\0}$.
Taylor expanding the respective \textit{eigenvalue functions} $g^t(\theta) = \cos^t(\theta)$ and $f_{t'}(\theta) = \cos(t'\theta)$ yields
\[
g^t(\theta) = 1 - \frac{t\theta^2}{2} + O(t^2\theta^4),
\quad \text{ whereas} \quad
f_{t'}(\theta) = 1 - \frac{t'^2\theta^2}{2} + O(t'^4\theta^4).
\]
Setting $t' = \sqrt{t}$, we see that both expressions are equal up to second order in $t$.
This suggests that the quantum walk \textit{quadratically fast-forwards} the Markov chain, and so $\Pr_\0 W^{\sqrt{t}} \approx \Pr_\0 \D^t$.

This observation underlies a range of quantum walk speedup results which are mainly concerned with accelerating the Markov chain asymptotics, where one is interested in the limit regime $\lim_{t\to\infty}\P^tv = \pi$ and one wishes to approximate the quantum state $\ket{\pi}$.
In these cases, the timescale for the classical Markov chain is for instance set by the inverse of the spectral gap $\frac{1}{\delta} = \frac{1}{1-\lambda_2}$ (for mixing tasks and Gibbs sampling, see \cite{somma2008quantum,poulin2009sampling}), or by the sum of the inverses $\sum \frac{1}{1-\lambda_k}$ (for hitting tasks, see \cite{szegedy2004quantum}), where $\{\lambda_k = \cos\theta_k\}$ denotes the set of eigenvalues of $P$.
For these purposes, the low order conclusions from the above expansion generally suffice to achieve a quantum walk speedup in generating $\ket{\pi}$.

The main issue with the above analysis is that it breaks down for $t$ and eigenvalues $\theta$ such that $t\theta \approx 1$: $g^t(\theta)$ and $f_t(\theta)$ start to diverge from each other, thus preventing the quantum walk from simulating the full dynamics of the Markov chain.
As the main contribution in the following section we will construct a more involved and fine-grained quantum walk scheme whose eigenvalue function closely approximates the Markov chain eigenvalue function $g^t(\theta)$ for all values of $t$ and $\theta$, without losing the quadratic fast-forward.

\subsection{Quantum Fast-Forwarding Algorithm}

In this section we develop our main tool: a quantum walk algorithm for quantum simulating Markov chains quadratically faster than the original dynamics.
Thereto we will make use of the concept of \textit{linear combinations of unitaries}.
We will use this technique to manipulate the eigenvalues of the quantum walk such that they better approximate the Markov chain eigenvalues.

\subsubsection{LCU and Chebyshev Expansion}
We can create some wiggle room on the implementation of the quantum walk $\Pr_\0 W^t$, and therefore on its eigenvalue function, by implementing linear combinations of $\Pr_\0 W^t$ for different $t$.
A similar approach has been used in for instance \cite{childs2012hamiltonian,berry2015simulating} for Hamiltonian simulation and in \cite{van2017quantum} for optimizing quantum SDP solvers, where they call this technique \textit{linear combination of unitaries} (LCU).
We extract the below Lemma \ref{lem:LCU} from this existing work, and elaborate its details for completeness.
Below the lemma we discuss how it can be used for our purpose.

The lemma shows how to implement a linear combination
\[
\sum_{l=0}^\tau q_l \Pr_\0 W^l,
\]
where we assume that $q_l \geq 0$ and $\sum q_l = 1$.
To do so, we will again enlarge the state space from $\Hi_{\h\V}$ to $\Hi_{\h\V \times [\tau]}$, with $[\tau] = \{0,1,2,\dots,\tau\}$.
We will identify $\ket{0} = \ket{\0}$ and call the span of states $\ket{j,\0,\0} \equiv \ket{j,\0\0}$, $j\in\V$, the flat subspace of $\Hi_{\h\V \times [\tau]}$.
The projector $\Pr_\0$ will either denote the projector on the flat subspace of $\Hi_{\h\V}$ or $\Hi_{\h\V \times [\tau]}$, whichever it is will be clear from the context.
The construction is very similar to the Watrous quantum walk scheme.
It builds on a coin toss $V_q$ on $\Hi_{\h\V \times [\tau]}$, defined by the coefficients $q_l$ as
\[
V_q \ket{\psi,\0}
= \sum_{l=0}^\tau \sqrt{q_l} \ket{\psi,l}.
\]
Then the \textit{controlled W-operator} $W_{\mathrm{ctrl}} = \sum_{l=0}^\tau W^l \otimes \ket{l}\bra{l}$ is applied which, conditioned on the integer $l$ in the last register, applies the operator $W^l$:
\[
W_{\mathrm{ctrl}} \, V_q \ket{\psi,\0}
= W_{\mathrm{ctrl}} \sum_{l=0}^\tau \sqrt{q_l} \ket{\psi,l}
= \sum_{l=0}^\tau \sqrt{q_l} W^l \ket{\psi,l}.
\]
Finally, as in the Watrous QW, the operator $V_q^\dag$ is applied, returning a state
\begin{equation} \label{eq:W-tau}
V_q^\dag W_{\mathrm{ctrl}} V_q \ket{\psi,\0}
= \sum_{l=0}^\tau q_l W^l \ket{\psi,\0} + \ket{\psi^\perp},
\end{equation}
where $\ket{\psi^\perp}$ is some quantum state perpendicular to the flat subspace.
This leads to the following lemma, where we set $W_\tau = V^\dag_q W_{\mathrm{ctrl}} V_q$.
\begin{lemma}[LCU] \label{lem:LCU}
For any $\ket{v,\0\0}$, it holds that
\[
\Pr_\0 W_\tau \ket{v,\0\0}
= \Big( \sum_{l=0}^\tau q_l \Pr_\0 W^l \ket{v,\0} \Big) \otimes \ket{\0}
= \Big( \sum_{l=0}^\tau q_l T_l(\D) \ket{v} \Big) \otimes \ket{\0\0}.
\]
Implementing the operator $W_\tau$ requires $O(\tau)$ QW steps.
\end{lemma}
\begin{proof}
From \eqref{eq:W-tau} we see that
\[
\Pr_\0 V_q^\dag W_{\mathrm{ctrl}} V_q \ket{v,\0\0}
= \Big( \sum_{l=0}^\tau q_l \Pr_\0 W^l \ket{v,\0} \Big) \otimes \ket{\0}.
\]
Combined with Proposition \ref{prop:QW-szegedy}, and by linearity, this proves the inequality.
In for instance \cite{nielsen2002quantum,berry2009black} it is discussed that the operator $W_{\mathrm{ctrl}}$ can be implemented in $O(\tau)$ QW steps, and the local coin tosses $V_q$ and $V_q^\dag$ require no QW steps.
\end{proof}

This lemma shows that if we apply the operator $W_\tau$ on $\ket{v,\0\0}$, and we perform a measurement $\{\Pr_{\0},\I-\Pr_{\0}\}$, then we retrieve the state $\Big( \sum_{l=0}^\tau q_l T_l(\D) \ket{v} \Big) \otimes \ket{\0\0}$ (up to normalization) with a probability $\|\sum_{l=0}^\tau q_l T_l(\D) \ket{v}\|^2$.
The corresponding eigenvalue function is then
\[
\tilde{f}_t(\theta)
= \sum_{l=0}^\tau q_l \cos(l\theta).
\]
In the following we will choose the coefficients $q_l$ so that $\tilde{f}_{t'}(\theta)$ approximates the original eigenvalue function $g^t(\theta) = \cos^t(\theta)$ for $t' \in O(\sqrt{t})$.
For this purpose we can use the Chebyshev expansion of $g^t$.
Indeed, from for instance \cite{gil2007numerical} we know that
\[
x^t
= \sum_{l=0}^t p_l T_l(x),
\]
where $p_l$ represents the probability that $|X_t| = l$ for $X_t$ a $t$ step random walk starting in the origin:
\begin{equation} \label{eq:coeff}
p_l
= \mathbb{P}(|X_t|=l)
= \begin{cases}
\frac{1}{2^{t-1}} \dbinom{t}{\frac{t-l}{2}} &\; l > 0,\, t = l \,\mathrm{mod}\, 2,\\
\frac{1}{2^t} \dbinom{t}{\frac{t}{2}} &\; l = 0,\, t = 0 \,\mathrm{mod}\, 2,\\
0 & \text{elsewhere}.
\end{cases}
\end{equation}
Again using the geometric definition of the Chebyshev polynomials, $T_t(\cos(\theta)) = \cos(t\theta)$, and setting $x = \cos(\theta)$, this implies that $g^t$ can be exactly expanded into the eigenfunctions $f_t$:
\begin{equation} \label{eq:cheb}
g^t(\theta) = \cos^t(\theta)
= \sum_{l=0}^t p_l \cos(l\theta)
= \sum_{l=0}^t p_l f_t(l\theta).
\end{equation}
Using the above lemma we can now choose $q_l=p_l$ to exactly simulate the original dynamics.
The problem is that in this case $\tau = t$, and implementing $W_\tau$ therefore requires $O(t)$ QW steps, giving no speedup with respect to the simple quantum simulation scheme.
We can resolve this by noting that $p_l$ approaches a normal distribution with standard deviation $\Theta(\sqrt{t})$, so that we can approximate it exponentially well by its support on a $O(\sqrt{t})$ interval, as we elaborate in the below lemma.

\begin{lemma} \label{lem:appr-cheb}
Let $\epsilon>0$.
If $C \geq 2\ln\frac{2}{\epsilon}$ then for all $\theta$
\[
\bigg| \cos^t(\theta) - \sum_{l=0}^{\lceil \sqrt{C t} \rceil} p_l \cos(l\theta) \bigg|
\leq \epsilon.
\]
\end{lemma}
\begin{proof}
Let $t' = \lceil \sqrt{Ct}\rceil$.
The proof comes down to bounding the quantity $p_{>t'} = \sum_{l=t'+1}^t p_l$.
Indeed, by \eqref{eq:cheb} we can easily calculate that
\[
\bigg| \cos^t(\theta) - \sum_{l=0}^{t'} p_l \cos(l\theta) \bigg|
\leq p_{>t'},
\]
so that it suffices to prove that $p_{>t'} \leq \epsilon$.
We can bound $p_{>t'}$ since it represents the probability that $|X_t|> t'$ where $X_t$ is a $t$ step random walk $X_t$.
By Hoeffding's inequality we know that $p_{>t'} \leq 2\exp\big(-t^{\prime 2}/(2t)\big)$.
For $t' = \lceil \sqrt{Ct}\rceil$ and $C \geq 2\ln\frac{2}{\epsilon}$ this shows that $p_{>\lceil\sqrt{Ct}\rceil}\leq\epsilon$, which proves the lemma.
\end{proof}

This lemma shows that it is possible to \textit{pointwise} approximate the original eigenvalue function $\cos^t(\theta)$, up to an error $\epsilon$, using the truncated Chebyshev expansion
\[
g^t_\tau(\theta)
= \sum_{l=0}^\tau p_l \cos(l\theta)
\]
for $\tau \in O(\sqrt{t}\log\frac{1}{\epsilon})$.
We note that a similar approximation in combination with LCU was used for a different purpose in \cite{berry2017exponential}.
In the next section we combine this approximation lemma with the LCU lemma, leading to our quantum fast-forwarding scheme.

\subsubsection{Quantum Fast-Forwarding Algorithm} \label{sec:QFF}
Combining the above lemmas, we can propose our QFF algorithm.
It builds on the operator $W_\tau$ from Lemma \ref{lem:LCU}, with coefficients $q_l$ derived from Lemma \ref{lem:appr-cheb}, so that
\begin{equation} \label{eq:W-tau-pl}
\begin{aligned}
\Pr_\0 W_\tau \ket{v,\0\0}
&= \frac{1}{1-p_{>\tau}} \sum_{l=0}^\tau p_l \Pr_\0 W^l \ket{v,\0\0} \\
&= \bigg( \frac{1}{1-p_{>\tau}} \sum_{l=0}^\tau p_l T_l(\D) \ket{v} \bigg) \otimes \ket{\0\0},
\end{aligned}
\end{equation}
where the $p_l$ are defined in \eqref{eq:coeff}.

\begin{algorithm}[H]
\caption{Quantum Fast-Forwarding $\mathbf{QFF}(\ket{v},\P,t,\epsilon)$} \label{alg:QFF}
\normalsize
\textbf{Input:} quantum state $\ket{v}\in\Hi_\V$, reversible Markov chain $\P$, $t\in\natural$, $\epsilon>0$ \\
\textbf{Do:}
\begin{algorithmic}[1]
\State set $\epsilon' = \|\D^t\ket{v}\|\epsilon/2$ and $\tau = \big\lceil \sqrt{2t\ln(2/\epsilon')} \big\rceil$
\State initialize the registers $\mathbf{R_1R_2R_3}$ with the state $\ket{v,\0\0}$
\State apply the LCU operator $W_\tau$ as in \eqref{eq:W-tau-pl} on $\mathbf{R_1R_2R_3}$
\State perform the measurement $\{\Pr_{\0\0},\I-\Pr_{\0\0}\}$
\State \textbf{if} outcome $\neq$ ``$\0\0$'' \textbf{then} output ``\texttt{Fail}'' and stop
\end{algorithmic}
\textbf{Output:} registers $\mathbf{R_1R_2R_3}$
\end{algorithm}

\begin{theorem}[Quantum Fast-Forwarding] \label{thm:QFF-success}
The QFF algorithm $\mathbf{QFF}(\ket{v},\P,t,\epsilon)$ outputs a state $\epsilon$-close to $\ket{\D^t v,\0\0}$ with success probability at least $(1-\epsilon) \|\D^t\ket{v}\|^2$.
Otherwise it outputs ``\texttt{Fail}''.
The algorithm uses a number of QW steps
\vspace{-.2cm}
\[
\tau
\in O\bigg( \sqrt{t} \, \log^{1/2} \frac{1}{\epsilon \|\D^t\ket{v}\|} \bigg).
\]
\end{theorem}
\begin{proof}
Let $\{(\cos\theta_k,\ket{v_k}), 1\leq k \leq |\V|\}$ be a complete orthonormal set of eigenpairs of $\D$.
Then we can write $\ket{v} = \sum_k \alpha_k \ket{v_k}$ and the goal state $\ket{\D^t v} = \sum_k \alpha_k \cos(\theta_k)^t\ket{v_k}/\|\D^t\ket{v}\|$.
From Lemma \ref{lem:LCU} we know that if we apply the operator $W_\tau$ on $\ket{v,\0\0}$, and we perform a measurement $\{\Pr_\0,\I-\Pr_\0\}$, then we retrieve the state
\begin{align*}
\frac{1}{\big\| \frac{1}{1-p_{>\tau}} \sum_{l=0}^\tau q_l T_l(\D) \ket{v} \big \|}
&\bigg( \frac{1}{1-p_{>\tau}} \sum_{l=0}^\tau q_l T_l(\D) \ket{v} \bigg) \otimes \ket{\0\0} \\
&= \frac{1}{\big\| \sum_{l=0}^\tau q_l T_l(\D) \ket{v} \big \|}
\bigg( \sum_{l=0}^\tau q_l T_l(\D) \ket{v} \bigg) \otimes \ket{\0\0}
\end{align*}
with a success probability $\|\frac{1}{1-p_{>\tau}} \sum_{l=0}^\tau p_lT_l(\D) \ket{v}\|^2$.
We will denote the state $\ket{\psi_\tau} = \sum_{l=0}^\tau q_l T_l(\D) \ket{v}$.
By the approximation from Lemma \ref{lem:appr-cheb} we know that if $\tau = \bigg\lceil \sqrt{2t\ln\frac{2}{\epsilon'}} \bigg\rceil$ then $\ket{\psi_\tau}$ will be $\epsilon'$-close to $\D^t\ket{v}$:
\begin{align*}
\left\| \ket{\psi_\tau} - \D^t \ket{v} \right\|
&= \sqrt{\sum_k \bigg| \sum_{l=0}^\tau q_l \cos(l\theta)
	- \cos^t(\theta_k) \bigg|^2 \cdot |\alpha_k|^2} \\
&\leq \sqrt{ \epsilon^{\prime 2} \smash[b]{\sum_k |\alpha_k|^2} } = \epsilon'.
\end{align*}
Using standard manipulations we know that for any two nonzero vectors it holds that $\| v/\|v\| - w/\|w\| \| \leq 2\|v-w\|/\|w\|$.
As a consequence we can bound
\[
\bigg\| \frac{\ket{\psi}}{\|\ket{\psi}\|} - \frac{\D^t\ket{v}}{\|\D^t\ket{v}\|} \bigg\|
\leq \frac{2\epsilon'}{\|\D^t\ket{v}\|} = \epsilon,
\]
using the fact that we chose $\epsilon' = \|\D^t\ket{v}\| \epsilon/2$.
We can now also bound the success probability using the reverse triangle inequality:
\begin{align*}
\bigg\| \frac{1}{1-p_{>\tau}} \sum_{l=0}^\tau p_lT_l(\D) \ket{v} \bigg\|^2
&\geq \bigg\| \sum_{l=0}^\tau p_lT_l(\D) \ket{v} \bigg\|^2 \\
&\geq ( \|D^t\ket{v}\| - \epsilon' )^2
\geq (1-\epsilon) \|\D^t\ket{v}\|^2.
\end{align*}
By Lemma \ref{lem:LCU} we know that implementing the operator $W_\tau$ requires a number of QW steps
\[
\tau
= \bigg\lceil \sqrt{2t} \ln^{1/2} \frac{4}{\epsilon\|\D^t\ket{v}\|} \bigg\rceil
\in O\bigg( \sqrt{t} \, \log^{1/2} \frac{1}{\epsilon \|\D^t\ket{v}\|} \bigg).
\]
This finalizes the proof.
\end{proof}

This theorem establishes our algorithm for quantum fast-forwarding Markov chains.
It winds back the quantum walk speedup of the Ambainis-Szegedy scheme on the Markov chain asymptotic behavior to the original problem of quantum simulating Markov chains for any fixed timestep, showing that we can achieve the same quadratic acceleration that is characteristic for this scheme.
The success probability is proportional to $\|D^t\ket{v}\|^2$, so that $\|D^t\ket{v}\|^{-2}$ expected iterations are necessary for the scheme to succeed.
In the next section we show how to quadratically improve this.
We mention that the norm $\|D^t\ket{v}\|$ will be small when the Markov chain spreads out from a small set to a large set, e.g., going from a single node to the uniform distribution yields $\|D^t\ket{v}\| = N^{-1/2}$.
This reflects the fact that it is costly for the quantum algorithm to create a superposition over a large number of queried elements (see \cite{ambainis2011symmetry} for a discussion and a more rigorous analysis of this phenomenon).

\subsubsection{Amplitude Amplification}

We can improve the success probability to a constant by replacing the final measurement in the algorithm with amplitude amplification, provided that we can reflect around the initial state $\ket{v,\0\0}$ ($\equiv \ket{v,\0}$, we will use the shorthand $\0$ to denote $\0\0$), i.e., implement the reflection operator
\[
R_v
= 2\ket{v,\0}\bra{v,\0} - \I.
\]
Thereto we will use the following proposition by Brassard et al \cite{brassard2002quantum}, demonstrating how we can retrieve a component $\Pr_\0\ket{\psi}$ of some quantum state $\ket{\psi}$ by performing reflections around $\ket{\psi}$ and around the image of $\Pr_\0$.
\begin{proposition}[Amplitude amplification \cite{brassard2002quantum}] \label{prop:ampl-ampl}
Given an initial state $\ket{\psi}$ and a projection operator $\Pr_\0$, with $\Pr_\0 \ket{\psi} \neq 0$.
Define the reflections $R_\psi = 2\ket{\psi}\bra{\psi} - \I$ and $R_\0 = 2\Pr_\0 - \I$, and set $m = \lfloor \pi/(4\theta) \rfloor$ with $\theta\in(0,\pi/2]$ such that $\sin\theta = \| \Pr_\0 \ket{\psi} \|$.
If we apply the operator $(-\R_\psi\R_\0)$ $m$ times on the state $\ket{\psi}$, and we perform a measurement $\{\Pr_\0,\I-\Pr_\0\}$, then we find back the state $\Pr_\0\ket{\psi}/\|\Pr_\0\ket{\psi}\|$ with probability at least $\max\{1-\sin^2\theta,\sin^2\theta\} \geq 1/2$.
\end{proposition}

This implementation of amplitude amplification requires a very good estimate of the initial success probability $\|\Pi_\0 \ket{\psi}\|^2$ to determine $m$.
If $m$ is chosen either too small or too large, then the guarantees on the success probability are lost, a problem often referred to as the ``souffl\'e problem''.
A remedy is however proposed in \cite{brassard2002quantum}, in which iteratively different guesses for $m$ are used.
They show that this approach also yields a success probability $\geq 1/2$, while still applying the operator $(-\R_\psi\R_\0)$ only $O(1/\|\Pi_\0 \ket{\psi}\|)$ times.
For clarity we will implicitly rely on this fact, and throughout assume that we can appropriately determine the parameter $m$.

In our QFF algorithm we have $\ket{\psi} = W_\tau \ket{v,\0}$, and we wish to retrieve the component $\Pr_\0 \ket{\psi} = \Pr_\0 W_\tau \ket{v,\0}$.
Amplitude amplification shows that we can do so with constant success probability by implementing the operator $(-\R_\psi\R_\0)$ for a number of times in $\Theta(1/\|\Pr_\0\ket{\psi}\|)$.
The reflection $\R_\0 = 2\Pr_\0 - \I$ is considered an elementary operation on the basis states, which we can implement with a negligible cost.
The following lemma shows that we can implement the reflection $\R_\psi = 2\ket{\psi}\bra{\psi} - \I$ using $O(\tau)$ QW steps.
\begin{lemma} \label{lem:cost-refl}
The operator $\R_\psi$ can be implemented using $O(\tau)$ QW steps and a reflection $\R_v$ around the initial state $\ket{v,\0}$.
\end{lemma}
\begin{proof}
We can rewrite the reflection $\R_\psi = 2\ket{\psi}\bra{\psi} - \I = W_\tau \R_v W_\tau^\dag$, so that we can implement the reflection by implementing the operators $W_\tau$, $W_\tau^\dag$, and $\R_v$.
To implement the operator $W_\tau^\dag$, we recall that $W_\tau = V_q^\dag \Big[ \sum_{l=0}^\tau \ket{l}\bra{l} \otimes (\R_\0\U)^l \Big] V_q$ and so $W_\tau^\dag = V_q^\dag \Big[ \sum_{l=0}^\tau \ket{l}\bra{l} \otimes (\U\R_\0)^l \Big] V_q$.
Here we also used the fact that $\U = V^\dag S V$ with $S = S^\dag$, as in \eqref{eq:QW-watrous}, so that $\U^\dag = \U$.
We already discussed in Lemma \ref{lem:LCU} that a controlled operator $\sum_{l=0}^\tau \ket{l}\bra{l}\otimes \U \R_\0$ can be implemented in $O(\tau)$ QW steps, so that both $W_\tau$ and $W_\tau^\dag$ can be implemented in $O(\tau)$ QW steps.
\end{proof}

\noindent
It follows that the total operator $(-\R_\psi\R_\0)$ can be implemented using $O(\tau)$ QW steps, a reflection around the initial state $\ket{v,\0}$, and a number of elementary operations.
In many cases the initial state will be an elementary basis state, so that the reflection $\R_v$ will also be elementary, and the main cost boils down to $O(\tau R)$ QW steps.
We can now propose the improved QFF algorithm, $\mathbf{QFFg}$, in Algorithm \ref{alg:QFFg}.
Theorem \ref{thm:QFFg} proves its correctness and complexity.
We note that this describes the Monte Carlo version of QFF.
We can easily retrieve the Las Vegas version, as stated in Theorem \ref{thm:QFFg-intro} in the introduction, by repeating the below algorithm until it is successful.
As mentioned at the end in previous section, we improve the expected number of QW steps for the QFF algorithm to succeed from $\widetilde{\Theta}(\|D^t \ket{v}\|^{-2} \sqrt{t})$ to $\widetilde{\Theta}(\|D^t \ket{v}\|^{-1} \sqrt{t})$.

\begin{algorithm} \label{alg:QFFg}
\caption{Quantum Fast-Forward with Reflections $\mathbf{QFFg}(\ket{v},\P,t,\epsilon)$} \label{alg:QFFg}
\normalsize
\textbf{Input:}
quantum state $\ket{v}\in\Hi_\V$, Markov chain $\P$, $t\in\natural$, $\epsilon>0$ \\
\textbf{Do:}
\begin{algorithmic}[1]
\State set $\epsilon' = \|\D^t\ket{v}\|\epsilon/2$ and $\tau = \Big\lceil \sqrt{2t\ln(2/\epsilon')} \Big\rceil$
\State set $m = \lfloor \pi/(4\theta)\rfloor$, where $0<\theta\leq \pi/2$ s.t.~$\sin\theta = \| \Pr_\0 W_\tau \ket{v} \|$
\State initialize registers $\mathbf{R_1R_2R_3}$ with the state $\ket{v,\0\0}$
\State apply the LCU operator $W_\tau$ on $\mathbf{R_1R_2R_3}$
\State apply the operator $(-\R_\psi\R_\0)^m = (-W_\tau \R_v W^\dag_\tau \R_\0)^m$
\hfill $\triangleright$ \textit{Amplitude Amplification}
\State perform the measurement $\{\Pr_{\0\0},\I-\Pr_{\0\0}\}$
\State \textbf{if} outcome $\neq$ ``$\0\0$'' \textbf{then} output ``\texttt{Fail}'' and stop
\end{algorithmic}
\textbf{Output:}
registers $\mathbf{R_1R_2R_3}$
\end{algorithm}

\begin{theorem}
\label{thm:QFFg}
The QFFg algorithm $\mathbf{QFFg}(\ket{v},\P,t,\epsilon)$ outputs a state $\epsilon$-close to $\ket{\D^t v}$ with success probability at least $1/2$.
Otherwise, it outputs ``\texttt{Fail}''.
The algorithm uses $\Theta(1/\|\D^t\ket{v}\|)$ reflections around the initial state, and a number of QW steps
\vspace{-.2cm}
\[
O(m\tau)
\in O\bigg(\frac{\sqrt{t}}{\|\D^t\ket{v}\|} \, \log^{1/2} \frac{1}{\epsilon \|\D^t\ket{v}\|} \bigg).
\]
\end{theorem}
\begin{proof}
The algorithm straightforwardly applies the amplitude amplification scheme on the state $W_\tau\ket{v,\0}$.
From Proposition \ref{prop:ampl-ampl} we know that the scheme has a success probability $\geq \max\{1-\sin^2\theta,\sin^2\theta \} \geq 1/2$.
The number of QW steps for implementing $W_\tau$ and $(-\R_\psi\R_\0)^m$ is $O(\tau)$ respectively $O(m\tau)$.
We know that $m \in O(1/\|\Pr_\0W_\tau\ket{v,\0}\|)$, and from the proof of Theorem \ref{thm:QFF-success} we can bound $\|\Pr_\0W_\tau\ket{v,\0}\| \geq \|\D^t\ket{v}\| - \epsilon' = (1-\epsilon/2)\|\D^t\ket{v}\| \in \Theta(\|\D^t\ket{v}\|)$ for all $\epsilon<1/2$.
\end{proof}

\subsection{Quantum Singular Value Transformation} \label{subsec:QSVT}

After completion of this work, we were pointed to the recent work of Gily\'en et al \cite{gilyen2018quantum} on quantum singular value transformation.
They develop a framework that generalizes and unifies the principles underlying a large number of quantum algorithms for problems such as Hamiltonian simulation, Gibbs sampling, and many more.
In the following we note that an alternative derivation of our QFF algorithm and its properties can be constructed from this framework.
Specifically, they consider a general projected unitary encoding of an operator $A = \Pr U \Pr'$, where $\Pr,\Pr'$ are projectors and $U$ is a unitary operator.
We can see the quantum walk encoding $D = \Pr_\0 U \Pr_\0$ of the discriminant matrix of a Markov chain, as in Lemma \ref{lem:QW-watrous}, as a special case of such encoding.
If $A$ has a singular value decomposition $A = W \Sigma V^\dag$, then they show that it is possible to transform the singular values of $A$.
More precisely, given some degree-$d$ polynomial $p$, they construct a quantum circuit that implements the operator $W p(\Sigma) V^\dag$ using the operators $U$ and $U^\dag$ at most $d$ times.
They then cite a result from Sachdeva and Vishnoi \cite{sachdeva2014faster} showing that if $p(\Sigma) = \Sigma^d$, then this polynomial can be efficiently approximated using a polynomial with degree $O(\sqrt{d})$ (this result also follows from our Chebyshev truncation in Lemma \ref{lem:appr-cheb}).
By applying their quantum singular value transformation on this approximating polynomial, we can retrieve our QFF algorithm.

\section{Quantum Property Testing} \label{sec:property-testing}

In this section we show how QFF allows to speed up random walk algorithms for property testing on graphs.
Specifically, we will consider property testers for the expansion and the clusterability of graphs.
We leave it as an open question whether similar speedups can be found for other graph property testers, an interesting example of which is the recent random walk algorithm by \cite{kumar2018fiding} for testing the occurrence of forbidden minors.
In the first section we will discuss the expansion tester of Goldreich and Ron (GR), which they presented in later work \cite{goldreich2011testing}, and we prove how QFF allows to accelerate this tester very naturally.
We compare this speedup to the preceding work by Ambainis, Childs and Liu \cite{ambainis2011quantum}, and discuss how they achieve a complementary speedup to ours.
Then we discuss the more recent line of algorithms for testing graph clusterability, aimed at probing for instance the community structure of a given graph.
We discuss how these testers make use of algorithms for classifying nodes in clusters, and show how QFF allows to accelerate these algorithms.
Such node classification is of relevance beyond the setting of property testing, allowing for instance nearest-neighbor classification of nodes in a learning problem.

\subsection{Classical Expansion Tester}
To formalize the concept of a graph property tester, we must introduce the notion of \textit{oracle or query access} to a graph as used throughout the literature on property testing.
Query access to an $N$-node undirected graph with degree bound $d$ means that we can query the graph with a string $(v,i) \in [N] \times [d]$, upon which we receive either the $i$-th neighbor $w \in [N]$ of $v$, or a special symbol in case $v$ has less than $i$ neighbors.
Clearly this model allows to perform a random walk over the graph.
In addition it allows to generate a uniformly random node by simply generating a random number in $[N]$.
This differs from the more classical Markov chain setting where possibly we are only given a single node, and we must explore the graph in a completely local manner.

Given such query access to a graph, the task of a property tester is to determine whether the graph has a certain property or is far from any graph having that property.
To formalize this, a distance measure between two $N$-node graphs $G$ and $G'$ is defined, equaling the number of edges that have to be added or removed from $G$ to transform it into $G'$.
With $\E$ and $\E'$ the edge sets of $G$ resp.~$G'$, this equals the size of the symmetric difference $|\E \triangle \E'|$.
We say that two $N$-node graphs $G$ and $G'$ with degree bound $d$ are $\epsilon$-far from each other if $|\E \triangle \E'| \geq \epsilon N d$.
Given a graph and a parameter $\epsilon$, a property tester should correctly distinguish between the graph having a certain property, and the graph being at least $\epsilon$-far from any graph with that property (i.e., the distance between the graph and any graph with that property is $\geq \epsilon N d$).
When given a graph that is neither, the algorithm can do whatever.

Goldreich and Ron \cite{goldreich2011testing} studied a property tester for the expansion of a graph\footnote{They actually studied a property tester for the spectral gap of a graph, for which currently there is no known sublinear algorithm. All follow-up work however considers the closely-related graph expansion.}.
The expansion or vertex expansion of a graph $G=(\V,\E)$ is defined by
\[
\min_{|\S| \leq |\V|/2} \frac{|\partial \S^c|}{|\S|},
\]
where $\partial \S^c$ denotes the outer boundary of $\S$, i.e., the set of nodes in $\S^c = \V\backslash\S$ that have an edge going to $\S$.
For some given parameter $\Upsilon$, an expansion tester should determine whether a given graph has expansion $\geq \Upsilon$, or whether it is at least $\epsilon$-far from any such graph.
GR, and the subsequent literature \cite{czumaj2010testing,kale2011expansion,nachmias2010testing,ambainis2011quantum}, have relaxed this setting somewhat.
They propose the following definition:
\begin{definition}
An algorithm is an $(\Upsilon,\epsilon,\mu)$-expansion tester if there exists a constant $c>0$, possibly dependent on $d$, such that given parameters $N$, $d$, and query access to an $N$-node graph with degree bound $d$ it holds that
\begin{itemize}
\item
if the graph has expansion $\geq \Upsilon$, then the algorithm outputs ``\texttt{accept}'' with probability at least $2/3$,
\item
if the graph is $\epsilon$-far from any graph having expansion $\geq c \mu \Upsilon^2$, then the algorithm outputs ``\texttt{reject}'' with probability at least $2/3$.
\end{itemize}
\end{definition}
\noindent
In the strict setting of property testing, the expression ``$\geq c \mu \Upsilon^2$'' in the second bullet should be replaced by ``$\geq \Upsilon$''.
Although unproven, the relaxation in this definition seems necessary to allow for efficient (sublinear) testing.
GR \cite{goldreich1997property} conjectured that the below Algorithm \ref{alg:get} is a $(\Upsilon,\epsilon,\mu)$ expansion tester.
They also proved that any classical expansion tester must make at least $\Omega(\sqrt{N})$ queries to the graph.

\begin{algorithm}
\caption{Graph Expansion Tester} \label{alg:get}
\normalsize
\textbf{Input:}
parameters $N$ and $d$; query access to $N$-node graph $G$ with degree bound $d$; expansion parameter $\Upsilon$; accuracy parameter $\epsilon$; running time parameter $\mu < 1/4$ \\
\textbf{Do:}
\begin{algorithmic}[1]
\For{$T \in \Theta(\epsilon^{-1})$ times}
\State
select a uniformly random starting vertex $s$
\State
perform $m \in \Theta(N^{1/2+\mu})$ independent random walks \newline
\-\ \hspace{5mm} of length $t \in \Theta(d^2 \Upsilon^{-2} \log N)$, starting in $s$
\State
count number of pairwise collisions between the endpoints of the $m$ random walks
\State
if the count is greater than $M \in \Theta(N^{2\mu})$, abort and output ``\texttt{reject}''
\EndFor
\end{algorithmic}
\textbf{Output:} if no ``\texttt{reject}'', output ``\texttt{accept}''
\end{algorithm}

The intuition behind the algorithm is very clear.
It builds on the use of a random walk $P$ on the given graph, which starting from a node $v$ jumps to any of its $d_v$ neighbors with a probability $1/(2d)$, and stays put otherwise:
\begin{equation} \label{eq:symm-RW}
P(u,v)
= \begin{cases}
1/(2d) & (v,u) \in E \\
1 - d_v/(2d) & u = v \\
0 & \text{elsewhere.} \\
\end{cases}
\end{equation}
This walk is lazy and symmetric, and hence converges to the uniform distribution.
If the graph has vertex expansion $\Upsilon$, then one can prove that the mixing time of this random walk is $O(d^2\Upsilon^{-2} \log N)$.
As a consequence, the probability distribution of the random walks in step 2 of the algorithm must be close to uniform.
If $p$ describes the probability distribution of the endpoint of a random walk starting from some fixed node, then the probability of a pairwise collision between two independent endpoints, i.e., the probability that the random walks end in the same node, is given by
\[
\sum p(j)^2
= \| p \|^2.
\]
This will be close to $1/N$ if $p$ is close to uniform.
If on the other hand the expansion of the graph is $\ll \Upsilon$, then random walks can get stuck in a small region, leading to an increase in the 2-norm or collision probability.
It follows that the collision probability of a random walk forms a measure for the expansion of the graph.
The key insight is then that, by a refinement of the birthday paradox, $\Theta(N^{1/2+\mu})$ independent samples of the same random walk suffice to estimate the collision probability to within a multiplicative factor $1 + O(N^{-2\mu})$.
As a consequence, it is possible to estimate the 2-norm of a $t$-step random walk on an $N$-node graph using $\Theta(N^{1/2+\mu} t \epsilon^{-1})$ random walk steps.

The conjecture that Algorithm \ref{alg:get} is an expansion tester was later resolved as the conclusion of a series of papers by Czumaj and Sohler \cite{czumaj2010testing}, Kale and Seshadhri \cite{kale2011expansion} and Nachmias and Shapira \cite{nachmias2010testing}, leading to the following theorem.

\begin{theorem}[\cite{nachmias2010testing}]
If $d \geq 3$, then Algorithm \ref{alg:get} is a $(\Upsilon,\epsilon,\mu)$ expansion tester with runtime
\[
O(N^{1/2+\mu} \Upsilon^{-2} d^2 \epsilon^{-1} \log N).
\]
\end{theorem}
\noindent
In the following section we show that we can use QFF to accelerate this tester very naturally by quantum simulating the random walks, and using quantum techniques to estimate the 2-norm.

\subsection{Quantum Expansion Tester} \label{sec:quantum-testing}

It is possible to extend the classical query model to the quantum setting, a proper definition of which can be found in \cite{montanaro2016survey,ambainis2011quantum}.
For this work it suffices to know that (i) we can generate a uniformly random node as in the classical case, and (ii) we can implement a single step of the quantum walk operator as defined in \eqref{eq:QW-ambainis} using $O(\sqrt{d})$ queries to the graph, where $d$ is the maximum degree.

To accelerate the classical tester we will quantum simulate the random walks, and then perform quantum amplitude estimation to estimate the 2-norm of the simulated random walks.
Together with the aforementioned amplitude amplification scheme, the amplitude estimation scheme is described in the work by Brassard et al \cite{brassard2002quantum}.
It is captured by the following lemma.
We note that in the original statement in \cite{brassard2002quantum} the number of reflections scales as $1/\delta$ for a success probability $1-\delta$.
We use standard tricks to improve this to $\log(1/\delta)$.
\begin{lemma}[Quantum Amplitude Estimation] \label{lem:QAE}
Consider a quantum state $\ket{\psi}$ and a general projector $\Pr_\0$.
Give some $\delta>0$, there exists a quantum algorithm that outputs an estimate $a$ such that $\big| \|\Pr_\0 \ket{\psi}\| - a \big| \leq \epsilon$ with probability at least $1-\delta$, using $O(\log(1/\delta)\epsilon^{-1})$ reflections around $\ket{\psi}$ and around the image of $\Pr_\0$.
\end{lemma}
\begin{proof}
We can use the quantum amplitude estimation algorithm from \cite[Theorem 12]{brassard2002quantum} to output an $\epsilon'=\epsilon/3$-close estimate with success probability at least $5/6$.
This algorithm requires $O(1/\epsilon)$ reflections around $\ket{\psi}$ and around the image of $\Pr_\0$.
We can boost the success probability to $1-\delta$ by running their algorithm $T = \lceil 18 \ln \delta^{-1} \rceil$ times, which by Hoeffding's inequality implies that, with a probability at least $1-\delta$, at least $2T/3$ iterations have been successful.
Therefore, with probability $1-\delta$, it holds that (i) at least $2T/3$ estimates lie in the interval $[\|\Pr_\0 \ket{\psi}\| - \epsilon',\|\Pr_\0 \ket{\psi}\| + \epsilon']$, and therefore (ii) we can find a subset $\S$ of estimates, $|\S| \geq 2T/3$, all of which lie in a $2\epsilon'$-interval.
This subset must overlap with the interval $[\|\Pr_\0 \ket{\psi}\| - \epsilon',\|\Pr_\0 \ket{\psi}\|+\epsilon']$.

If now we output any element of this subset $\S$, we know that it lies in the interval $[\|\Pr_\0 \ket{\psi}\| - 3\epsilon',\|\Pr_\0 \ket{\psi}\| + 3\epsilon']$.
This proves that with probability $1-\delta$ we can output an estimate of $\|\Pr_\0 \ket{\psi}\|$ with precision $3\epsilon' = \epsilon$, using $T$ runs of the quantum amplitude estimation algorithm in \cite{brassard2002quantum}, each of which requires $O(1/\epsilon)$ reflections around $\ket{\psi}$ and around the image of $\Pr_\0$.
This proves the claimed statement.
\end{proof}

We will use this amplitude estimation algorithm to estimate the 2-norm of a random walk.
Thereto we recall the QFF scheme as discussed in Section \ref{sec:QFF}.
Note that the random walk \eqref{eq:symm-RW} proposed in the GR tester is symmetric, so that we can simply replace the discriminant matrix $D$ in the QFF algorithm by the random walk matrix $P$.
Given a quantum state $\ket{s,\0\0}$, QFF applies an operator $W_\tau$ so that
\[
\Pr_\0 W_\tau \ket{s,\0\0}
= \bigg( \frac{1}{1-p_{>\tau}} \sum_{l=0}^\tau p_l T_l(P) \ket{s} \bigg) \otimes \ket{\0\0}
\approx (P^t \ket{s}) \otimes \ket{\0\0},
\]
as in \eqref{eq:W-tau-pl}, with the summation corresponding to the truncated Chebyshev expansion of $P^t$.
Implementing this operator requires $O(\tau)$ QW steps and $O(\tau \sqrt{d})$ queries to the graph.
If we set $\tau \in \Theta\big(\sqrt{t \log (N\epsilon^{-1})}\big)$ (replacing $\|P^t\ket{s}\|$ by its lower bound $N^{-1/2}$ in Algorithm \ref{alg:QFFg}) then the 2-norm of $\frac{1}{1-p_{>\tau}} \sum_{l=0}^\tau p_l T_l(P) \ket{v}$ approximates the 2-norm of $P^t \ket{v}$ to precision $O(\epsilon)$.
Applying quantum amplitude estimation on the state $W_\tau \ket{v,\0\0}$ and projector $\Pr_\0$ will therefore allow to estimate the 2-norm of $P^t \ket{v}$, as was our initial goal.
This scheme is easily formalized in the below algorithm and theorem.

\begin{algorithm}[H]
\caption{Quantum 2-norm Estimator} \label{alg:ql2t}
\normalsize
\textbf{Input:} parameters $N$ and $d$; query access to $N$-node graph $G$ with degree bound $d$; starting vertex $s$; running time $t$; accuracy parameter $\epsilon$; confidence parameter $\delta$ \\
\textbf{Do:}
\begin{algorithmic}[1]
\State
set $\tau \in O(\sqrt{t \log (N/\epsilon)})$
\State
apply the QFF operator $W_\tau$ on the quantum state $\ket{s,\0\0}$
\State
use quantum amplitude estimation to create estimate $a$ of $\|\Pr_\0 W_\tau \ket{s,\0\0}\|$ \newline
\-\ to error $\epsilon/2$ with probability $1-\delta$
\end{algorithmic}
\textbf{Output:} estimate $a$

\end{algorithm}

\begin{theorem}[Quantum 2-norm Estimator] \label{thm:ql2t}
With probability at least $1-\delta$, Algorithm \ref{alg:ql2t} outputs an estimate $a$ such that $\big| \| P^t \ket{s} \| - a \big| \leq \epsilon$.
The algorithm requires a number of QW steps bounded by $O\big( \frac{\sqrt{t}}{\epsilon} \log\frac{1}{\delta} \log^{1/2} \frac{N}{\epsilon} \big)$.
\end{theorem}
\begin{proof}
We will prove the theorem for the algorithm parameter $\tau = \Big\lceil \sqrt{2t \ln\big(8\sqrt{N}/\epsilon\big)} \Big\rceil$.
By this choice, and the fact that $\|P^t\ket{s}\| \geq N^{-1/2}$ on any $N$-node graph, we can deduce from the proof of Theorem \ref{thm:QFF-success} that
\[
\big| \|\Pr_\0 W_\tau \ket{s}\| - \|P^t \ket{s}\| \big|
\leq \epsilon/2.
\]
Applying quantum amplitude estimation on $\Pr_\0 W_\tau \ket{s}$ with a precision $\epsilon/2$ therefore leads to an estimate of $\|P^t \ket{s}\|$ up to error $\epsilon$.
By Lemma \ref{lem:QAE} we can do so with a probability $1-\delta$ using $O(\epsilon^{-1} \log \delta^{-1})$ reflections around $W_\tau \ket{s}$.
We can implement a single such reflection using $2\tau$ QW steps, and a reflection around the initial state (which is a basis state and can hence be neglected).
\end{proof}

We can compare this with the classical 2-norm tester proposed by Czumaj, Peng and Sohler \cite[Lemma 3.2]{czumaj2015testing}, building on the GR tester.
For $\delta = 1/3$ their tester requires $O(t/\epsilon)$ queries to the graph, whereas our algorithm only requires $\widetilde{O}(\sqrt{t}/\epsilon)$ queries.
We can now use our quantum 2-norm tester to create a quantum tester for the graph expansion.
The proof makes use of some details from the classical proof of Nachmias and Shapira \cite{nachmias2010testing}.

\begin{algorithm}[H]
\caption{Quantum Graph Expansion Tester} \label{alg:qget}
\normalsize
\textbf{Input:} parameters $N$ and $d$; query access to $N$-node graph $G$ with degree bound $d$; expansion parameter $\Upsilon$; accuracy parameter $\epsilon$; running time parameter $\mu < 1/4$ \\
\textbf{Do:}
\begin{algorithmic}[1]
\State
set $t \in O(d^2 \Upsilon^{-2} \log N)$, $M \in O(N^{-1/2})$, $\epsilon' \in O(N^{-1/2+\mu})$, $\delta \in O(\epsilon)$
\For{$T \in O(\epsilon^{-1})$ times}
\State
select a uniformly random starting node $s$
\State
use Algorithm \ref{alg:ql2t} to create estimate $a$ of $\| P^t \ket{s} \|$ to precision $\epsilon'$, \newline
\-\ \hspace{5mm} with probability $1-\delta$
\State
if $a > M + \epsilon'$, abort and output ``\texttt{reject}''
\EndFor
\end{algorithmic}
\textbf{Output:} if no ``\texttt{reject}'', output ``\texttt{accept}''
\end{algorithm}

\begin{theorem}[Quantum Graph Expansion Tester] \label{thm:QGET}
If $d \geq 3$ then Algorithm \ref{alg:qget} is a $(\Upsilon,\epsilon,\mu)$ expansion tester.
The runtime and number of queries of the algorithm are bounded by
\[
O(N^{1/2+\mu} d^{3/2} \Upsilon^{-1} \epsilon^{-1} \log(\epsilon^{-1}) \log N).
\]
\end{theorem}
\begin{proof}
We will prove the theorem for the algorithm parameters $t = 16 d^2 \Upsilon^{-2} \log N$, $M = \sqrt{N^{-1}(1+N^{-1})}$, $\epsilon' = N^{-1/2+\mu}/(16\sqrt{2})$, $\delta = \epsilon/300$ and $T = 90/\epsilon$.

In each iteration the estimate $a$ will be such that $\big| a - \|P^t\ket{s}\| \big| \leq \epsilon'$ with probability $1-\delta$, and hence
\[
\big| a^2 - \|P^t\ket{s}\|^2 \big|
= \big| (a - \|P^t\ket{s}\|)(a + \|P^t\ket{s}\|) \big|
\leq 2 \epsilon'.
\]
Nachmias and Shapira \cite{nachmias2010testing} showed that if $G$ has a conductance $\geq \Upsilon$, then for all nodes $s$ it holds that
\[
\| P^t \ket{s} \|
\leq M = \sqrt{N^{-1} (1 + N^{-1})}.
\]
Given such a graph, in each iteration the estimate $a \leq M + \epsilon'$ with probability $1-\delta$, so that the probability of a faulty rejection is at most $\delta$ per iteration.
The total probability of a faulty rejection can then be bounded by $T\delta < 1/3$.

In the negative case, \cite{nachmias2010testing} showed that there exists a constant $c>0$ such that if $G$ is $\epsilon$-far from any graph with max degree $d$ and conductance $\geq c\mu\Upsilon^2$, then there exist at least $\epsilon N/128$ vertices $s$ for which it holds that
\[
\| P^t \ket{s} \|
\geq \sqrt{N^{-1} (1 + 32^{-1} N^{-2\mu} )}.
\]
Given such a graph, in each iteration the estimate $a$ will now be such that $a \geq \|P^t\ket{s}\| - \epsilon' \geq \sqrt{N^{-1} (1 + 32^{-1} N^{-2\mu} )} - \epsilon'$ with probability $1-\delta$.
To show that this quantity $> M + \epsilon'$, we bound $M = \sqrt{N^{-1} (1 + N^{-1})} \leq N^{-1/2} (1+N^{-1/2})$ and $\sqrt{N^{-1} (1 + 32^{-1} N^{-2\mu} )} \geq N^{-1/2} (1 - N^{-\mu}/(4\sqrt{2}))$, which shows that
\[
\sqrt{N^{-1}(1+32^{-1} N^{-2\mu} )} - M
\geq N^{-1/2-\mu}/(4\sqrt{2}).
\]
This proves that indeed
\[
\sqrt{N^{-1} (1 + 32^{-1} N^{-2\mu} )} - \epsilon'
> M+\epsilon'\]
for $\epsilon' = N^{-1/2-\mu}/(16\sqrt{2})$.
As a consequence, a single iteration will correctly output a rejection with probability $(1-\delta)\epsilon N/128$.
The total probability of correctly rejecting at least once is therefore lower bounded by $T(1-\delta) \epsilon/128 \geq 2/3$.
This concludes the proof that Algorithm \ref{alg:qget} is a $(\Upsilon,\epsilon,\mu)$ graph expansion tester.
The required number of QW steps is given by $T$ times the number required by the 2-norm tester, which by Theorem \ref{thm:ql2t} is
\[
O\bigg( \frac{\sqrt{t}}{\epsilon'} \log\frac{1}{\delta} \log^{1/2} \frac{N}{\epsilon'} \bigg)
\in O\bigg( \big( d\Upsilon^{-1} \log^{1/2} N \big)
		 N^{1/2+\mu} \log \frac{1}{\epsilon} \log^{1/2} N^{1+\mu} \bigg).
\]
We can implement a single QW step using $O(\sqrt{d})$ graph queries, so that we find the claimed bound.
\end{proof}
\noindent
We recall that the classical GR tester has a runtime
\[
O(N^{1/2+\mu} d^2 \Upsilon^{-2} \epsilon^{-1} \log N).
\]
Up to the $\log(\epsilon^{-1})$-factor we improve this runtime with a factor $\Upsilon^{-1}$, which basically follows from the fact that we quadratically accelerate the random walk runtime to $t \in \widetilde{\Theta}(\Upsilon^{-1})$.
There exist bounded-degree graphs for which $\Upsilon \in \Omega(1/N)$, so that in some cases we improve the runtime by a factor $\Upsilon^{-1} \in \Theta(N)$.
Concerning the space complexity, we note that the classical GR tester must store and compare the endpoints of $\Omega(N^{1/2})$ independent random walks.
By direct inspection we see that our algorithm only requires $O(\log(N t \log\epsilon^{-1})) \in \mathrm{polylog}(N)$ qubits to implement, exponentially improving the space complexity.
This is due to the fact that our algorithm compares superpositions that encode the endpoint distribution of the random walks, rather than an explicit list of samples.

We can now compare this result to the preceding work by Ambainis, Childs and Liu \cite{ambainis2011quantum}.
They used a very different approach to speed up the GR expansion tester, using quantum walks only indirectly, which results in a runtime improvement of a different nature.
In rough strokes they speed up the classical 2-norm tester by making use of Ambainis' quantum walk algorithm for element distinctness \cite{ambainis2007quantum} to count collisions between pairs of classical random walks more efficiently.
This allows them to improve the runtime of the 2-norm tester to $\widetilde{O}(N^{1/3+\mu} t)$, which provides a speedup complementary to the speedup of our 2-norm tester which in this context has a runtime $\widetilde{O}(N^{1/2+\mu} \sqrt{t})$.
Using this 2-norm tester in the above Algorithm \ref{alg:qget} leads to a runtime
\[
\widetilde{O}(N^{1/3+\mu} d^2 \Upsilon^{-2} \epsilon^{-1}).
\]
The space complexity of this approach is comparable to the GR tester: the algorithm for element distinctness over the $\sqrt{N}$ random walk endpoints requires to store $N^{1/3}$ elements.
We leave it as future work to combine the $\tilde{\Theta}(N^{1/6})$ gain in collision counting of \cite{ambainis2011quantum} with our $\tilde{\Theta}(d\Upsilon^{-1})$ gain in random walk runtime and our logarithmic space complexity.

We note that a property tester in the same spirit as the GR expansion tester was proposed by Batu et al \cite{batu2013testing} for testing the mixing time of general Markov chains on a graph.
For the special case of symmetric Markov chains it seems feasible that we can speed up their algorithm using the same ideas, yielding a similar speedup on the random walk runtime.

\subsection{Clusterability and Robust $s$-$t$ Connectivity} \label{sec:cluster}

We can similarly use QFF to speed up a more recent line of algorithms on testing the clusterability of a graph \cite{czumaj2015testing,chiplunkar2018testing}.
In clusterability testing the goal is to test whether a graph can be appropriately clustered into $k$ parts for some given $k$.
The proposed algorithms build on a subroutine of independent interest, which allows to determine whether a pair of nodes lie in the same cluster or not.
This leads to a robust notion of $s$-$t$ connectivity, useful e.g.~for classifying objects among a set of examples and relevant also outside of the setting of property testing.
We show that QFF allows to speed up this subroutine, leading to a speedup on the clusterability testers that use this subroutine.

The observation underlying the clusterability testers in \cite{czumaj2015testing,chiplunkar2018testing} is that the GR technique of counting collision can also be used to estimate the inproduct of any two given distributions $p$ and $q$, defined by
\[
\braket{p,q}
= \sum_j p(j) q(j).
\]
Indeed, this quantity is equal to the collision probability between the two distributions.
The estimate on the inproduct of the inproduct is then used to estimate the 2-distance $\|p-q\|$ between a pair of random walks, which will be small if both random walks started in the same cluster, and large otherwise.
This approach of estimating the distance between distributions was further developed in the work by Batu et al \cite{batu2013testing}, Valiant \cite{valiant2011testing} and Chan et al \cite{chan2014optimal}.
We will focus our efforts on showing how QFF allows to speed up this routine of independent interest, following up with an informal discussion of how this leads to a speedup on the clusterability tester of Czumaj et al \cite{czumaj2015testing}.

\subsubsection*{2-distance Estimator}
To estimate the 2-distance of a pair of random walks, we will combine QFF with the SWAP test: given two quantum states $\ket{\psi}$ and $\ket{\phi}$ and an ancillary qubit in the state $\ket{0}$, yielding the state $\ket{0} \ket{\psi} \ket{\phi}$, we apply the following operations:
\[
\begin{array}{r c l}
\ket{0} \ket{\psi} \ket{\phi}
& \overset{H \otimes I \otimes I}{\rightarrow}
& \frac{\ket{0} + \ket{1}}{\sqrt{2}} \ket{\psi} \ket{\phi} \\
& \overset{cS}{\rightarrow}
& \frac{1}{\sqrt{2}} \ket{0} \ket{\psi} \ket{\phi} + \frac{1}{\sqrt{2}} \ket{1} \ket{\phi} \ket{\psi} \\
& \overset{H \otimes I \otimes I}{\rightarrow}
& \frac{1}{2} \ket{0} (\ket{\psi} \ket{\phi} + \ket{\phi} \ket{\psi})
+ \frac{1}{2} \ket{1} (\ket{\psi} \ket{\phi} - \ket{\phi} \ket{\psi}),
\end{array}
\]
where we used the Hadamard gate $H = \frac{1}{\sqrt{2}} \begin{bmatrix} 1 & 1 \\ 1 & -1 \end{bmatrix}$, and the \textit{conditional swap operation} $cS$ swapping the second and third registers conditional on the first register being in the state $\ket{1}$.
We will call the combined unitary operation $U_{\mathrm{SWAP}} = (H \otimes I \otimes I) cS (H \otimes I \otimes I)$.
We can now either measure the first register, or apply quantum amplitude estimation to the projector $\Pr_1 = \ket{1}\bra{1} \otimes I \otimes I$, to estimate the quantity
\[
\| \ket{\psi} \ket{\phi} - \ket{\phi} \ket{\psi} \|^2
= 2(1 - |\braket{\psi|\phi}|^2).
\]
This quantity will be small if $\ket{\psi}$ and $\ket{\phi}$ are close and large otherwise, allowing to estimate the distance between the input states $\ket{\psi}$ and $\ket{\phi'}$.
We can combine the SWAP test with our QFF algorithm, and the 2-norm estimator in previous section, to obtain a tester for the 2-distance.
Due to the straightforward yet technical nature of the details of the tester, we defer its description to Appendix \ref{app:2-dist-est}.

\begin{theorem}[Quantum 2-distance Estimator] \label{thm:2-distance}
With probability at least $1-\delta$, Algorithm \ref{alg:q2ld} outputs an estimate $a$ such that
\[
\big| \|P^t\ket{u} - P^t\ket{v}\|^2 - a \big|
\leq \epsilon.
\]
For $\amax = \max\{\|P^t\ket{u}\|,\|P^t\ket{v}\|\}$ and $\amin = \min\{\|P^t\ket{u}\|,\|P^t\ket{v}\|\}$, the algorithm requires an expected number of QW steps bounded by
\[
O\left( \sqrt{t} \bigg( \frac{\amax}{\epsilon} + \frac{\amax^4}{\amin \epsilon^2} \bigg)
			\log \frac{\log N}{\delta} \log^{3/2} \frac{N}{\epsilon} \right).
\]
\end{theorem}

\noindent
For comparison, the classical estimator presented in Czumaj et al \cite[Theorem 3.1]{czumaj2015testing} requires a number of graph queries or random walk steps $O\big( t \frac{\amax}{\epsilon} \log\frac{1}{\delta} \big)$.
Chan et al \cite{chan2014optimal} give an information theoretical proof that classically $\Omega(\amax/\epsilon)$ samples are needed to estimate the 2-distance between a pair of distributions.

\subsubsection*{Classifying Nodes}
Czumaj et al \cite{czumaj2015testing} use their classical 2-distance estimator to propose a property tester for the clusterability of a graph.
Following for instance Oveis Gharan and Trevisan \cite{gharan2014partitioning}, they say that a graph $G$ is \textit{$(k,\phin,\phout)$-clusterable} if and only if there exists a partition $\V = \S_1 \cup \dots \cup \S_h$, $h\leq k$, such that the clusters are well-connected internally, $\Phi(G[\S_j]) \geq \phin$, and poorly-connected externally, $\Phi(\S_j) \leq \phout$.
Here $G[\S_j]$ denotes the graph consisting of the nodes in $\S_j$ and the edges between these nodes, the conductance $\Phi(\S_j)$ is defined as
\[
\Phi(\S_j)
= \frac{|E(\S_j,\S_j^c)|}{d|\S_j|},
\]
and the conductance $\Phi(G')$ of a graph $G'=(\V',\E')$ is
\[
\Phi(G')
= \min_{\T\subset\V', |\T|\leq |\V'|/2} \frac{|E(\T,\V'\backslash \T)|}{d|\T|}.
\]
It turns out that graph clusterability can be efficiently tested when the gap between $\phin$ and $\phout$ is sufficiently large - typically quadratic, $\phout \in \widetilde{O}(\phin^2)$.

Czumaj et al \cite{czumaj2015testing} construct such a clusterability tester using a subroutine for \textit{classifying nodes}, i.e., determining whether two nodes lie in the same cluster or not.
As mentioned before, it is possible to classify nodes by comparing random walks starting from the nodes: the 2-distance between random walks starting from nodes of the same cluster will typically be smaller than the 2-distance between nodes from different clusters.
This is formalized below in Lemma \ref{lem:class-nodes}, which we extract from Czumaj et al \cite[Lemma 4.1 and 4.3]{czumaj2015testing}.
Given an appriopriately clusterable graph, having a gap $\phout \in O(\phin^2/\log N)$, it gives bounds on the 2-distance between pairs of nodes coming from the same or different clusters.
The lemma is confined to the \textit{internal nodes} $\tilde{\S} \subseteq \S$ of a cluster $\S$, similar to most work on locally exploring graph clusters, see for instance the work of Spielman and Teng \cite{spielman2004nearly}.

\begin{lemma}[\cite{czumaj2015testing}] \label{lem:class-nodes}
Consider a $(k,\phin,\phout)$-clusterable graph with degree bound $d$, and let $\S$ and $\S'$ be clusters of such a partition.
Assume that
\[
\phout \leq c \phin^2/\log N,
\]
with $c$ some constant dependent on $d$, $k$, $|\S|/N$ and $|\S'|/N$.
Then there exist subsets $\tilde{\S} \subseteq \S$, $|\tilde{\S}| \geq |\S|/2$, and $\tilde{\S'} \subseteq \S'$, $|\tilde{\S'}| \geq |\S'|/2$, and a universal constant $c'$, such that for $t = \lceil c' k^4 \phin^{-2} \log N \rceil$ it holds that
\begin{itemize}
\item
if two nodes $u,v \in \tilde{\S}$ or $u,v \in \tilde{\S'}$, then $\| P^t\ket{u} - P^t\ket{v} \|^2 \leq 1/(4N)$.
\item
if two nodes $u \in \tilde{\S}$ and $v \in \tilde{\S'}$, then $\| P^t\ket{u} - P^t\ket{v} \|^2 \geq 1/N$.
\end{itemize}
\end{lemma}

\noindent
We can combine this lemma with our quantum 2-distance estimator to prove the below proposition.
It speeds up the routine which lies at the basis of the property tester in \cite{czumaj2015testing}, which essentially solves a robust version of $s$-$t$ connectivity.
Arguably the latter is more relevant to e.g.~social networks, where mere connectivity between two nodes is no longer deemed an interesting quantity, yet the community or cluster structure does hold important information.

\begin{proposition}[Classifying Nodes] \label{prop:class-nodes}
\begin{itemize}
\item
Under the clusterability conditions of Lemma \ref{lem:class-nodes}, we can use the quantum 2-distance estimator to determine with probability at least $2/3$ whether two internal nodes lie in the same cluster or not.
\item
There exists a subset $\tilde{\V} \subseteq \V$, $|\tilde{\V}| \geq 9|\V|/10$, such that if in addition both nodes lie in $\tilde{\V}$, then the algorithm requires $O(N^{1/2} k^4 \phin^{-1} \log^{3/2} N)$ expected QW steps.
\end{itemize}
\end{proposition}
\begin{proof}
To prove the first bullet, it suffices to use Lemma \ref{lem:class-nodes} which states that if both lie in the same cluster, then $\| P^t\ket{u} - P^t\ket{v} \|^2 \leq 1/(4N)$, whereas if both lie in different clusters, then $\| P^t\ket{u} - P^t\ket{v} \|^2 \geq 1/N$.
By Theorem \ref{thm:2-distance} we can estimate $\| P^t\ket{u} - P^t\ket{v} \|^2$ to error $\epsilon = 1/N$, which allows to distinguish both cases.

To prove the second bullet, let $\tilde{\V}$ denote a set of nodes $u$ for which $\|P^t\ket{u}\| \in O(k/N)$, which by \cite[Lemma 4.2]{czumaj2015testing} we know we can choose of size at least $9|\V|/10$.
If both nodes lie in $\tilde{\V}$, then in Theorem \ref{thm:2-distance} we can set $\amax \in O(k/N)$, and $\amin \in O(1/N)$ since necessarily $\|P^t\ket{u}\| \geq 1/N$ for any node $u$.
In this case, the expected number of QW steps becomes $O(\sqrt{t N} \log^{3/2} N)$.
For $t$ as in Lemma \ref{lem:class-nodes}, this proves the second bullet.
\end{proof}

\noindent
We can compare the runtime in the second bullet by the runtime when using classical collision counting, which requires a number of RW steps $\widetilde{O}(N^{1/2} k^4 \phin^{-2})$.
Applying the element distinctness technique by Ambainis et al \cite{ambainis2011quantum} requires a number of QW steps $\widetilde{O}(N^{1/3} k^4 \phin^{-2})$.
Again we also find an improvement in space complexity with respect to these alternative approaches: our algorithm only requires $\mathrm{polylog}(N)$ qubits, whereas the other approaches require $\mathrm{poly}(N)$ classical or quantum bits.

Lemma \ref{lem:class-nodes}, combined with a classifier as in Proposition \ref{prop:class-nodes}, forms the basis of the graph clusterability tester proposed by \cite{czumaj2015testing}.
Since the tester is in the same vein as the GR expansion tester, we will not state it explicitly but merely summarize the idea.
The algorithm selects a uniformly random set of $\Theta(k\log k)$ nodes over which it constructs a \textit{similarity graph} by adding an edge between any pair of nodes if their random walk probabilities are closer than some threshold.
This similarity graph serves as a \textit{graph sketch}, reminiscent of the recent surge of results on graph sketching and sparsification \cite{batson2013spectral}.
They then prove that if the graph is appropriately clusterable in at most $k$ components, then with high probability this small similarity graph will have at most $k$ connected components, which they then check by brute force.
Using the classical 2-distance estimator to estimate the distance between random walk distributions, this leads to a clusterability tester requiring $\widetilde{O}\big(N^{1/2} k^7 \phin^{-2} \epsilon^{-5}\big)$ RW steps.
We can improve this to $\widetilde{O}\big(N^{1/2} k^7 \phin^{-1} \epsilon^{-4}\big)$ QW steps using Proposition \ref{prop:class-nodes}.
It seems feasible that using the element distinctness technique in \cite{ambainis2011symmetry} an alternative speedup to $\widetilde{O}\big(N^{1/3} k^7 \phin^{-2} \epsilon^{-5}\big)$ RW steps can be achieved.

\section{Discussion and Open Questions}

We introduced a new quantum walk tool called quantum fast-forwarding (QFF), allowing to quantum simulate classical reversible Markov chains with a quadratically improved time dependency.
The main benefit of this tool is that it allows to effectively simulate the transient dynamics of the Markov chains.
We can contrast this to many existing quantum walk algorithms which rely on a speedup of the Markov chain limit behavior.
This new feature is crucial for the applications in graph property testing and node classification that we discuss.
Indeed we show that QFF allows to speed up in a very natural way random walk algorithms for testing graph properties such as expansion and clusterability, both of which decisively depend on the transient dynamics of a random walk.

To finalize we mention some avenues for future work:
\begin{itemize}
\item
\textit{Improving the QFF scheme: parameter dependence and irreversible Markov chains.}
QFF allows to create an $\epsilon$-approximation of the state $\ket{\D^tv}$ with constant success probability using a number of QW steps
\[
O\bigg(\frac{\sqrt{t}}{\|\D^t\ket{v}\|} \log^{1/2}\frac{1}{\epsilon\|\D^t\ket{v}\|}\bigg)
\]
and $O(\|\D^t\ket{v}\|^{-1})$ reflections around the initial state $\ket{v}$.
It is easy to see that the individual $t$ and $\epsilon$ dependency are optimal by looking at the random walk on $\integer$.
If we tolerate an $\epsilon$ error, then we can confine the probability distribution of a $t$-step random walk to the $\Theta\big(t^{1/2} \log^{1/2} \epsilon^{-1}\big)$ neighborhood of the initial state.
Since the QW has the same locality constraints as the RW, it needs $\Omega\big(t^{1/2} \log^{1/2} \epsilon^{-1}\big)$ QW steps to spread out over this interval.
A very similar argument also shows why in general QFF cannot create the state $\ket{\P^t v}$ (rather than $\ket{\D^t v}$) when $\P$ is irreversible.
Indeed, consider the Markov chain on $\integer$ which simply moves to the right every step, $\P(i+1,i) = 1$ and $\P(i-1,i) = 0$.
This walk is clearly not reversible, as the direction of its motion reverses when running the time forward or backward.
When starting in the origin, the walk will be on node $t$ after $t$ steps.
A local QW requires $\Omega(t)$ steps to reach this point, so that no fast-forwarding is possible.

We leave improvements of the dependency on $\|D^t \ket{v}\|$ as an open question.

\item
\textit{Local Graph Clustering and Sparsification.}
Local graph clustering algorithms, as in \cite{spielman2004nearly,andersen2009finding}, aim to explicitly construct a local cluster, rather than merely test whether appropriate clusters exist.
They have a similar flavor to the graph expansion tester that we discussed, making use of random walks and other diffusive dynamics as a way of locally exploring a graph.
It might be possible to use QFF or similar ideas as a way of speeding up these algorithms.
Since these algorithms formed the root of a number of approaches towards graph sparsification and solving symmetric diagonally-dominant linear systems, this might lead to speedups on these highly relevant problems as well.

\item
\textit{Hamiltonian QFF.}
Following for instance Childs \cite{childs2010relationship}, we can associate a QW to a general Hermitian matrix, representing for instance a Hamiltonian rather than a Markov chain.
We leave it as an open question whether this can lead to interesting applications in for instance imaginary time evolution \cite{verstraete2008matrix} or an improved implementation of functions of a Hamiltonian \cite{childs2017quantum}.
\end{itemize}

\bibliography{/home/simon/Dropbox/biblio}

\begin{thebibliography}{10}
\providecommand{\url}[1]{\texttt{#1}}
\providecommand{\urlprefix}{URL }

\bibitem{ambainis2007quantum}
Ambainis, A.: Quantum walk algorithm for element distinctness. SIAM Journal on
  Computing  37(1),  210--239 (2007)

\bibitem{magniez2007quantum}
Magniez, F., Santha, M., Szegedy, M.: Quantum algorithms for the triangle
  problem. SIAM Journal on Computing  37(2),  413--424 (2007)

\bibitem{childs2003exponential}
Childs, A.M., Cleve, R., Deotto, E., Farhi, E., Gutmann, S., Spielman, D.A.:
  Exponential algorithmic speedup by a quantum walk. In: Proceedings of the
  35th Annual ACM Symposium on Theory of Computing. pp. 59--68. ACM (2003)

\bibitem{szegedy2004quantum}
Szegedy, M.: Quantum speed-up of {M}arkov chain based algorithms. In:
  Proceedings of the 45th Annual IEEE Symposium on Foundations of Computer
  Science. pp. 32--41. IEEE (2004)

\bibitem{krovi2016quantum}
Krovi, H., Magniez, F., Ozols, M., Roland, J.: Quantum walks can find a marked
  element on any graph. Algorithmica  74(2),  851--907 (2016)

\bibitem{ambainis2001one}
Ambainis, A., Bach, E., Nayak, A., Vishwanath, A., Watrous, J.: One-dimensional
  quantum walks. In: Proceedings of the 33rd Annual ACM Symposium on Theory of
  Computing. pp. 37--49. ACM (2001)

\bibitem{aharonov2001quantum}
Aharonov, D., Ambainis, A., Kempe, J., Vazirani, U.: Quantum walks on graphs.
  In: Proceedings of the 33rd Annual ACM Symposium on Theory of Computing. pp.
  50--59. ACM (2001)

\bibitem{richter2007quantum}
Richter, P.C.: Quantum speedup of classical mixing processes. Physical Review A
   76(4),  042306 (2007)

\bibitem{somma2008quantum}
Somma, R.D., Boixo, S., Barnum, H., Knill, E.: Quantum simulations of classical
  annealing processes. Physical Review Letters  101(13),  130504 (2008)

\bibitem{wocjan2008speedup}
Wocjan, P., Abeyesinghe, A.: Speedup via quantum sampling. Physical Review A
  78(4),  042336 (2008)

\bibitem{aharonov2003adiabatic}
Aharonov, D., Ta-Shma, A.: Adiabatic quantum state generation and statistical
  zero knowledge. In: Proceedings of the 35th Annual ACM Symposium on Theory of
  Computing. pp. 20--29. ACM (2003)

\bibitem{ambainis2003quantum}
Ambainis, A.: Quantum walks and their algorithmic applications. International
  Journal of Quantum Information  1(04),  507--518 (2003)

\bibitem{santha2008quantum}
Santha, M.: Quantum walk based search algorithms. In: International Conference
  on Theory and Applications of Models of Computation. pp. 31--46. Springer
  (2008)

\bibitem{watrous2001quantum}
Watrous, J.: Quantum simulations of classical random walks and undirected graph
  connectivity. Journal of Computer and System Sciences  62(2),  376--391
  (2001)

\bibitem{aleliunas1979random}
Aleliunas, R., Karp, R.M., Lipton, R.J., Lovasz, L., Rackoff, C.: Random walks,
  universal traversal sequences, and the complexity of maze problems. In:
  Proceedings of the 20th Annual IEEE Symposium on Foundations of Computer
  Science. pp. 218--223. IEEE (1979)

\bibitem{childs2012hamiltonian}
Childs, A.M., Wiebe, N.: Hamiltonian simulation using linear combinations of
  unitary operations. Quantum Information and Computation  12(11\& 12),
  901--924 (2012)

\bibitem{berry2015simulating}
Berry, D.W., Childs, A.M., Cleve, R., Kothari, R., Somma, R.D.: Simulating
  hamiltonian dynamics with a truncated taylor series. Physical Review Letters
  114(9),  090502 (2015)

\bibitem{van2017quantum}
Van~Apeldoorn, J., Gily{\'e}n, A., Gribling, S., de~Wolf, R.: Quantum
  {SDP}-solvers: better upper and lower bounds. In: Proceedings of the 58th
  Annual IEEE Symposium on Foundations of Computer Science. pp. 403--414. IEEE
  (2017)

\bibitem{diaconis2013spectral}
Diaconis, P., Miclo, L.: On the spectral analysis of second-order {M}arkov
  chains. Annales de la Facult{\'e} des Sciences de Toulouse. Math{\'e}matiques
   22(3),  573--621 (2013)

\bibitem{gilyen2018quantum}
Gily{\'e}n, A., Su, Y., Low, G.H., Wiebe, N.: Quantum singular value
  transformation and beyond: exponential improvements for quantum matrix
  arithmetics. arXiv:1806.01838  (2018)

\bibitem{spielman2013local}
Spielman, D.A., Teng, S.H.: A local clustering algorithm for massive graphs and
  its application to nearly linear time graph partitioning. SIAM Journal on
  Computing  42(1),  1--26 (2013)

\bibitem{goldreich2011testing}
Goldreich, O., Ron, D.: On testing expansion in bounded-degree graphs. In:
  Studies in Complexity and Cryptography. Miscellanea on the Interplay between
  Randomness and Computation, pp. 68--75. Springer (2011)

\bibitem{ambainis2011quantum}
Ambainis, A., Childs, A.M., Liu, Y.K.: Quantum property testing for
  bounded-degree graphs. In: Approximation, Randomization, and Combinatorial
  Optimization. Algorithms and Techniques, pp. 365--376. Springer (2011)

\bibitem{czumaj2015testing}
Czumaj, A., Peng, P., Sohler, C.: Testing cluster structure of graphs. In:
  Proceedings of the 47th Annual ACM Symposium on Theory of Computing. pp.
  723--732. ACM (2015)

\bibitem{chiplunkar2018testing}
Chiplunkar, A., Kapralov, M., Khanna, S., Mousavifar, A., Peres, Y.: Testing
  graph clusterability: Algorithms and lower bounds. In: Proceedings of the
  59th Annual IEEE Symposium on Foundations of Computer Science. IEEE (2018)

\bibitem{valiant2011testing}
Valiant, P.: Testing symmetric properties of distributions. SIAM Journal on
  Computing  40(6),  1927--1968 (2011)

\bibitem{magniez2011search}
Magniez, F., Nayak, A., Roland, J., Santha, M.: Search via quantum walk. SIAM
  Journal on Computing  40(1),  142--164 (2011)

\bibitem{poulin2009sampling}
Poulin, D., Wocjan, P.: Sampling from the thermal quantum gibbs state and
  evaluating partition functions with a quantum computer. Physical Review
  Letters  103(22),  220502 (2009)

\bibitem{nielsen2002quantum}
Nielsen, M.A., Chuang, I.: Quantum computation and quantum information.
  Cambridge University Press (2002)

\bibitem{berry2009black}
Berry, D.W., Childs, A.M.: Black-box hamiltonian simulation and unitary
  implementation. Quantum Information and Computation  12(1\& 2),  29--62
  (2012)

\bibitem{gil2007numerical}
Gil, A., Segura, J., Temme, N.M.: Numerical methods for special functions,
  vol.~99. SIAM (2007)

\bibitem{berry2017exponential}
Berry, D.W., Childs, A.M., Cleve, R., Kothari, R., Somma, R.D.: Exponential
  improvement in precision for simulating sparse hamiltonians. In: Forum of
  Mathematics, Sigma. vol.~5. Cambridge University Press (2017)

\bibitem{ambainis2011symmetry}
Ambainis, A., Magnin, L., Roetteler, M., Roland, J.: Symmetry-assisted
  adversaries for quantum state generation. In: Proceedings of the 26th Annual
  Conference on Computational Complexity. pp. 167--177. IEEE (2011)

\bibitem{brassard2002quantum}
Brassard, G., H{\oe}yer, P., Mosca, M., Tapp, A.: Quantum amplitude
  amplification and estimation. Contemporary Mathematics  305,  53--74 (2002)

\bibitem{sachdeva2014faster}
Sachdeva, S., Vishnoi, N.K., et~al.: Faster algorithms via approximation
  theory. Foundations and Trends{\textregistered} in Theoretical Computer
  Science  9(2),  125--210 (2014)

\bibitem{kumar2018fiding}
Kumar, A., Seshadhri, C., Stolman, A.: Fiding forbidden minors in sublinear
  time: a $o(n^{1/2 + o(1)})$-query one-sided tester for minor closed
  properties on bounded degree graphs. arXiv:1805.08187  (2018)

\bibitem{czumaj2010testing}
Czumaj, A., Sohler, C.: Testing expansion in bounded-degree graphs.
  Combinatorics, Probability and Computing  19(5-6),  693--709 (2010)

\bibitem{kale2011expansion}
Kale, S., Seshadhri, C.: An expansion tester for bounded degree graphs. SIAM
  Journal on Computing  40(3),  709--720 (2011)

\bibitem{nachmias2010testing}
Nachmias, A., Shapira, A.: Testing the expansion of a graph. Information and
  Computation  208(4),  309 (2010)

\bibitem{goldreich1997property}
Goldreich, O., Ron, D.: Property testing in bounded degree graphs. In:
  Proceedings of the 29th Annual ACM Symposium on Theory of Computing. pp.
  406--415. ACM (1997)

\bibitem{montanaro2016survey}
Montanaro, A., de~Wolf, R.: A survey of quantum property testing. Theory of
  Computing Library Graduate Surveys (7),  1--81 (2016)

\bibitem{batu2013testing}
Batu, T., Fortnow, L., Rubinfeld, R., Smith, W.D., White, P.: Testing closeness
  of discrete distributions. Journal of the ACM (JACM)  60(1), ~4 (2013)

\bibitem{chan2014optimal}
Chan, S.O., Diakonikolas, I., Valiant, P., Valiant, G.: Optimal algorithms for
  testing closeness of discrete distributions. In: Proceedings of the 25th
  Annual ACM-SIAM Symposium on Discrete Algorithms. pp. 1193--1203. SIAM (2014)

\bibitem{gharan2014partitioning}
Gharan, S.O., Trevisan, L.: Partitioning into expanders. In: Proceedings of the
  25th Annual ACM-{SIAM} {Symposium} on {Discrete} {Algorithms}. pp.
  1256--1266. SIAM (2014)

\bibitem{spielman2004nearly}
Spielman, D.A., Teng, S.H.: Nearly-linear time algorithms for graph
  partitioning, graph sparsification, and solving linear systems. In:
  Proceedings of the 36th Annual ACM symposium on Theory of Computing. pp.
  81--90. ACM (2004)

\bibitem{batson2013spectral}
Batson, J., Spielman, D.A., Srivastava, N., Teng, S.H.: Spectral sparsification
  of graphs: theory and algorithms. Communications of the ACM  56(8),  87--94
  (2013)

\bibitem{andersen2009finding}
Andersen, R., Peres, Y.: Finding sparse cuts locally using evolving sets. In:
  Proceedings of the 41st Annual ACM Symposium on Theory of Computing. pp.
  235--244. ACM (2009)

\bibitem{childs2010relationship}
Childs, A.M.: On the relationship between continuous-and discrete-time quantum
  walk. Communications in Mathematical Physics  294(2),  581--603 (2010)

\bibitem{verstraete2008matrix}
Verstraete, F., Murg, V., Cirac, I.J.: Matrix product states, projected
  entangled pair states, and variational renormalization group methods for
  quantum spin systems. Advances in Physics  57(2),  143--224 (2008)

\bibitem{childs2017quantum}
Childs, A.M., Kothari, R., Somma, R.D.: Quantum algorithm for systems of linear
  equations with exponentially improved dependence on precision. SIAM Journal
  on Computing  46(6),  1920--1950 (2017)

\bibitem{yoder2014fixed}
Yoder, T.J., Low, G.H., Chuang, I.L.: Fixed-point quantum search with an
  optimal number of queries. Physical Review Letters  113(21),  210501 (2014)

\end{thebibliography}

\begin{appendices}

\section{Quantum 2-distance Estimator: Algorithm and Proof} \label{app:2-dist-est}

In this appendix we present the algorithm and proof underlying Theorem \ref{thm:2-distance}, which concerns the estimation of the distance between two random walk distributions $p = P^t \ket {u}$ and $q = P^t \ket{v}$.
To construct our algorithm, we rewrite
\[
\| p - q \|^2
= \|p\|^2 + \|q\|^2 - 2\|p\|\|q\| \braket{p|q},
\]
using the notation $\braket{p|q} = \braket{p,q}/(\|p\| \|q\|)$.
As a consequence, we can retrieve an estimate by separately estimating $\|p\|$, $\|q\|$ and $\braket{p|q}$.
Towards estimating $\|p\|$ and $\|q\|$, we present at the end of this appendix a simple extension of the quantum 2-norm tester presented earlier in Section \ref{sec:quantum-testing} that allows to estimate the 2-norm up to multiplicative error, instead of additive error.
Towards estimating $\braket{p|q}$, we first create approximations of $\ket{p} = p/\|p\|$ and $\ket{q} = q/\|q\|$, on which we subsequently apply the SWAP test and amplitude estimation.
A subtlety is that we cannot simply use our QFF algorithm to create $\ket{p}$ and $\ket{q}$ with high probability.
Indeed, in order to apply amplitude estimation for the SWAP test we must reflect around these states, and it is not clear that we can reflect around the output of the QFF algorithm.
Instead, we will apply the unitary amplitude amplification operator to the states $W_\tau \ket{u,\0\0}$ and $W_\tau \ket{v,\0\0}$ to unitarily rotate these states close to $\ket{p}$ and $\ket{q}$, omitting the final measurement in Algorithm \ref{alg:QFFg}.
This invertible operation will allow to reflect around the output states.
Furthermore, instead of the amplitude amplification operator used in Section \ref{sec:QFF}, we will make use of an enhanced operator by Yoder and Low \cite{yoder2014fixed}.
This operator, as described in the below lemma, is better suited for the case where we only have a lower bound on the success probability.

\begin{lemma}[Fixed Point Amplitude Amplification \cite{yoder2014fixed}] \label{lem:FPAA}
Consider a state $\ket{\psi}$ and a projector $\Pr_\0$ such that $\| \Pr_\0 \ket{\psi} \| = \lambda > 0$.
For any constant $\delta>0$, there exists a family of unitary transformations $U_L$ such that if $L\geq \lambda^{-1} \log(2/\delta)$ then
\[
| \braket{\psi_\0|U_L|\psi} |^2
\geq 1 - \delta^2,
\]
where $\psi_\0 = \Pr_\0 \ket{\psi}/\| \Pr_\0 \ket{\psi} \|$.
We can implement $U_L$ using $O(L)$ reflections around $\ket{\psi}$ and around the image of $\Pr_\0$.
\end{lemma}

Using the appropriate operator $U_L$, we can therefore retrieve approximations $\ket{\psi_u} = U_L W_\tau \ket{u,\0\0} \approx \ket{p}$ and $\ket{\psi_v} = U_L W_\tau \ket{v,\0\0} \approx \ket{q}$.
We can now apply the SWAP test to these states, combined with amplitude amplification, to retrieve an estimation of $\braket{p|q}$.
To see this, note that
\[
\Pr_1 (U_{\mathrm{SWAP}}\ket{0} \ket{\psi_u} \ket{\psi_v} )
= \frac{1}{2} \ket{1} (\ket{\psi_u}\ket{\psi_v} - \ket{\psi_v}\ket{\psi_u}).
\]
As a consequence we can apply quantum amplitude estimation on the state $U_{\mathrm{SWAP}}\ket{0} \ket{\psi_u} \ket{\psi_v}$ with respect to the projector $\Pr_1$ to estimate the quantity
\[
\frac{1}{2} \big\| \ket{\psi_u}\ket{\psi_v} - \ket{\psi_v}\ket{\psi_u} \big\|^2 \\
= 1 - |\braket{\psi_u|\psi_v}|^2
\approx 1 - |\braket{p|q}|^2.
\]
Combined with the former estimates of $\|p\|$ and $\|q\|$ this leads to an estimate of the 2-distance we were looking for.
We formalize this in the following algorithm and theorem.

\begin{algorithm}
\caption{Quantum 2-distance Estimator} \label{alg:q2ld}
\normalsize
\textbf{Input:} parameters $N$ and $d$; query access to $N$-node graph $G$ with degree bound $d$; starting vertices $u$ and $v$; running time $t$; accuracy parameter $\epsilon$; confidence parameter $\delta$ \\
\textbf{Do:}
\begin{algorithmic}[1]
\State
use Algorithm \ref{alg:q2lt-m} to create estimates $\alpha$ and $\beta$ of $\|P^t\ket{u}\|$ resp.~$\|P^t\ket{v}\|$ \newline
\-\ to multiplicative error $1/4$, with probability $1-\delta/4$
\State
set $\mu \in O(\epsilon \max(\alpha,\beta)^{-2})$
\State
use Algorithm \ref{alg:q2lt-m} to create new estimates $\alpha$ and $\beta$ of $\|P^t\ket{u}\|$ resp.~$\|P^t\ket{v}\|$ \newline
\-\ to multiplicative error $\mu$, with probability $1-\delta/4$
\State
set $L \in \Omega(\min(\alpha,\beta)^{-1} \log \min(\alpha,\beta)^{-1})$ and $\tau \in \Omega(\sqrt{t \ln(N/\mu)})$
\State
apply $W_\tau$, $U_L$ and $U_{\mathrm{SWAP}}$ to create the state
\[
\ket{\psi}
= U_{\mathrm{SWAP}} \ket{0} \big( U_LW_\tau\ket{u,\0\0} \big) \big( U_LW_\tau\ket{v,\0\0} \big)
\]
\State
use amplitude estimation to create an estimate $\gamma$ of $\| \Pr_1 \ket{\psi} \|$ to error $\mu$, \newline
\-\ with probability $1-\delta/2$
\end{algorithmic}
\textbf{Output:} estimate $a = \alpha^2 + \beta^2 - 2\alpha\beta \sqrt{1-\gamma^2/2}$
\end{algorithm}

\begin{theorem}[Quantum 2-distance Estimator]
With probability at least $1-\delta$, Algorithm \ref{alg:q2ld} outputs an estimate $a$ such that
\[
\big| \|P^t\ket{u} - P^t\ket{v}\|^2 - a \big|
\leq \epsilon.
\]
With $\amax = \max\{\|P^t\ket{u}\|,\|P^t\ket{v}\|\}$ and $\amin = \min\{\|P^t\ket{u}\|,\|P^t\ket{v}\|\}$, the algorithm requires an expected number of QW steps bounded by
\[
O\left( \sqrt{t} \bigg( \frac{\amax}{\epsilon} + \frac{\amax^4}{\amin \epsilon^2} \bigg)
			\log \frac{\log N}{\delta} \log^{3/2} \frac{N}{\epsilon} \right).
\]
\end{theorem}
\begin{proof}
We prove the theorem for
\[
\mu
= \frac{1}{26} \min\Big( 1, \frac{9\epsilon}{16\max(\alpha,\beta)^2} \Big), \quad
L
= \Big\lceil \frac{1}{\lambda} \log \frac{2}{\nu} \Big\rceil, \quad
\tau
= \Big\lceil \sqrt{2t} \ln^{1/2} \frac{4}{\lambda \nu} \Big\rceil,
\]
with $\lambda = \min(\alpha,\beta)/(1+\nu)$ and $\nu = \mu^2/11$.
We will denote $p = P^t\ket{u}$, $q = P^t\ket{v}$, $\ket{p} = p/\|p\|$, $\ket{q} = q/\|q\|$, $\amax^2 = \max(\|p\|,\|q\|)$ and $\amin = \min(\|p\|,\|q\|)$.
The algorithm estimates the quantity $\| p - q\|^2 = \|p\|^2 + \|q\|^2 - 2\|p\| \|q\| \braket{p|q}$ by separately estimating $\|p\|$, $\|q\|$ and $\braket{p|q}$ to error $O(\epsilon/\amax^2)$.

After the first step, we retrieve with probability at least $1-\delta/4$ estimates $\alpha$ and $\beta$ such that
\[
\frac{3}{4} \|p\| \leq \alpha \leq \frac{5}{4} \|p\|, \qquad
\frac{3}{4} \|q\| \leq \beta \leq \frac{5}{4} \|q\|.
\]
This proves that the parameter
\begin{equation} \label{eq:bound-mu}
\mu
= \frac{1}{26} \min\bigg(1,\frac{\epsilon}{(4\max(\alpha,\beta)/3)^2} \bigg)
\leq \frac{1}{26} \min\bigg(1,\frac{\epsilon}{\amax^2} \bigg),
\end{equation}
and $\mu \in \Theta(\min(1,\epsilon/\amax^2))$.
In step 3 we then create new estimates of $\|p\|$ and $\|q\|$ to multiplicative error $\mu$.
The combined success probability of both steps is $(1-\delta/4)^2 \geq 1-\delta/2$.
Following Theorem \ref{thm:ql2t-m} these steps require an expected number of QW steps in
\[
O\left( \frac{\sqrt{t} \amax}{\epsilon} \log \frac{\log N}{\delta}
			\log^{1/2} \frac{N}{\epsilon} \right).
\]

In the following steps of the algorithm we estimate $\braket{p|q} = \frac{\braket{p,q}}{\|p\| \|q\|}$ to additive error $\mu$ by combining QFF, amplitude amplification, the SWAP test and amplitude estimation.
Thereto we first rewrite
\[
\braket{p|q}
= \sqrt{1 - \frac{\| \ket{p}\ket{q} - \ket{q}\ket{p} \|^2}{2}},
\]
showing that we can use an estimate on $\| \ket{p}\ket{q} - \ket{q}\ket{p} \|$ to estimate $\braket{p|q}$.
Indeed, it is easily seen from a function plot that if we create an estimate $\kappa \in [0,\sqrt{2}]$ such that $\big| \| \ket{p}\ket{q} - \ket{q}\ket{p} \| - \kappa \big| \leq \mu^2$, then the estimate $\sqrt{1 - \smash[b]{\kappa^2/2}}$ will be $\mu$-close:
\begin{equation} \label{eq:ineq-nu}
\Big| \sqrt{1 - \smash[b]{\kappa^2/2}} - \braket{p|q} \Big|
\leq \mu.
\end{equation}
We now create an estimate of $\| \ket{p}\ket{q} - \ket{q}\ket{p} \|$.
By Lemma \ref{lem:FPAA} and Theorem \ref{thm:QFFg}, and our choice of $L$ and $\tau$, it holds that
\begin{align*}
\| U_L W_\tau \ket{u,\0\0} - \ket{p,\0\0} \|
&\leq (1-\nu^2) \| W_\tau \ket{u,\0\0}/\|W_\tau \ket{u,\0\0}\| - \ket{p,\0\0} \| + \nu \\
&\leq (1-\nu^2) \nu + \nu \leq 2\nu,
\end{align*}
with $\nu = \mu^2/11$, and similarly for $U_L W_\tau \ket{v,\0\0}$.
If we set $\ket{\psi_u} = U_L W_\tau \ket{u,\0\0}$ and $\ket{\psi_v} = U_L W_\tau \ket{v,\0\0}$, then this implies that
\[
\big| \| \ket{\psi_u}\ket{\psi_v} - \ket{\psi_v}\ket{\psi_u} \|
		 - \| \ket{p}\ket{q} - \ket{q}\ket{p} \| \big|
\leq 8\nu (1 + 2\nu).
\]
Now we can apply amplitude estimation, as in Lemma \ref{lem:QAE}, to the state $U_{\mathrm{SWAP}} \ket{0}\ket{\psi_u}\ket{\psi_v}$ and projector $\Pr_1$ with success probability $1-\delta/2$ and error $\nu$.
If successful this returns an estimate $\gamma$ of $\| \ket{\psi_u}\ket{\psi_v} - \ket{\psi_v}\ket{\psi_u} \|$ to error $\nu$.
Combined with the above inequality this shows that
\[
\big| \| \ket{p}\ket{q} - \ket{q}\ket{p} \| - \gamma \big|
\leq \nu + 8\nu (1 + 2\nu)
\leq \mu^2.
\]
By \eqref{eq:ineq-nu} this leads to the promised bound $\big| \sqrt{1 - \smash[b]{\gamma^2/2}} - \braket{p|q} \big| \leq \mu$.

Implementing $W_\tau$, $U_L$ and $U_{\mathrm{SWAP}}$ requires a number of QW steps $O(\tau) + O(L)$, bounded by
\[
O\left( \frac{\sqrt{t}}{\amin} \log \frac{\amax}{\epsilon\amin}
			\log^{1/2} \frac{N \amax}{\epsilon} \right).
\]
Applying amplitude estimation with success probability $1-\delta/2$ and error $\nu \in \Theta(\epsilon^2/\amax^4)$ requires $O\big(\frac{\amax^4}{\epsilon^2} \log \frac{1}{\delta}\big)$ reflections around the state $U_{\mathrm{SWAP}} \ket{0}\ket{\psi_u}\ket{\psi_v}$.
We can implement each such reflection using the same number of QW steps required to implement the operators $W_\tau$, $U_L$ and $U_{\mathrm{SWAP}}$.
This leads to a total number of QW steps bounded by
\[
O\left( \frac{\sqrt{t} \amax^4}{\amin \epsilon^2} \log \frac{1}{\delta} \log \frac{\amax}{\epsilon\amin}
			\log^{1/2} \frac{N \amax}{\epsilon} \right).
\]

Combined with the first approximation part, we find estimates $\alpha$, $\beta$ and $\gamma$ such that $|\alpha - \|p\|| \leq \mu\|p\|$, $|\beta - \|q\|| \leq \mu \|q\|$ and $|\gamma - \braket{p|q}| \leq \mu$.
This allows to prove the claimed error of the estimate
\begin{align*}
\big| \alpha^2 + \beta^2 - 2\alpha \beta \gamma - \| p - q \|^2 \big|
&\leq \mu(2+\mu)(\|p\|^2 + \|q\|^2) \\
&\qquad + 2\|p\| \|q\| \big[\mu(2+\mu) (\braket{p|q}+\mu) + (1+\mu)^2 \mu \big] \\
&\leq 3\mu(\|p\|^2 + \|q\|^2) + 20\mu \|p\| \|q\| \\
&\leq 26 \mu \max(\|p\|,\|q\|)^2
\leq \epsilon,
\end{align*}
using the bound \eqref{eq:bound-mu}.
The total success probability can be bounded by $(1-\delta/2)^2 \geq 1 - \delta$, and the expected number of QW steps by
\[
O\left( \sqrt{t} \bigg( \frac{\amax}{\epsilon} + \frac{\amax^4}{\amin \epsilon^2} \bigg)
			\log \frac{\log N}{\delta} \log^{3/2} \frac{N}{\epsilon} \right).
\]
\end{proof}

\subsection*{2-norm Estimator to Multiplicative Error}

In the above estimator for the 2-distance we wish to estimate $\| P^t\ket{u}\|$ to some multiplicative error $\epsilon$, without having a bound on $\| P^t\ket{u}\|$.
We present such an estimator in the below algorithm and theorem.

\begin{algorithm}
\caption{Quantum Multiplicative 2-norm Estimator} \label{alg:q2lt-m}
\normalsize
\textbf{Input:} parameters $N$ and $d$; query access to $N$-node graph $G$ with degree bound $d$; starting vertex $u$; running time $t$; accuracy parameter $\epsilon$; confidence parameter $\delta$ \\
\textbf{Do:}
\begin{algorithmic}[1]
\For{$k = 1 \dots T \in O(\log N)$}
\State
use Algorithm \ref{alg:ql2t} to create estimate $\alpha$ of $\|P^t\ket{u}\|$ to error $\epsilon_k = \epsilon 2^{-k-2}$,\newline
\-\ \hspace{5mm} with probability $1-\delta'$ for $\delta' \in O(\delta \log^{-1} N)$
\State
if $\alpha \geq (1+\epsilon) 2^{-k}$, abort \textbf{for}-loop
\EndFor
\end{algorithmic}
\textbf{Output:} $\alpha$
\end{algorithm}

\begin{theorem}[Quantum Multiplicative 2-norm Estimator] \label{thm:ql2t-m}
With probability at least $1-\delta$, Algorithm \ref{alg:q2ld} outputs an estimate $\alpha$ such that
\[
\big| \|P^t\ket{u}\| - \alpha \big|
\leq \epsilon \|P^t\ket{u}\|.
\]
The algorithm requires an expected number of QW steps bounded by
\[
O\left( \frac{\sqrt{t}}{\epsilon \|p\|} \log \frac{\log N}{\delta} \log^{1/2} \frac{N}{\epsilon} \right).
\]
\end{theorem}
\begin{proof}
We will prove the theorem for $T = \lceil \frac{1}{2} \log N \rceil$ and $\delta' = \delta/T$.
We do so by showing that with probability at least $1-\delta$ the loop aborts such that the value of $\alpha$ forms an estimate of $\|p\|$ to multiplicative error $\epsilon$, where we denote $p = P^t\ket{u}$.
We first assume that every call to Algorithm \ref{alg:ql2t} is successful, the probability of which we will bound afterwards.
Let $a_k$ be the value of $\alpha$ in the $k$-th iteration, so that $|\|p\| - a_k| \leq \epsilon_k$.
If the loop is stopped at the $k$-th iteration then $a_k \geq (1+\epsilon) 2^{-k}$ or equivalently $\epsilon_k \leq \frac{\epsilon}{1+\epsilon} a_k$.
Combined with the fact that $a_k \leq \|p\| + \epsilon_k$ this shows that $\epsilon_k \leq \frac{\epsilon}{1+\epsilon} (\|p\| + \epsilon_k)$ or equivalently $\epsilon_k \leq \epsilon \|p\|$, so that we find an estimate with multiplicative error $\epsilon$.

If the first $\big\lceil \log \|p\|^{-1} \big\rceil$ calls to the 2-norm estimator are successful, then the algorithm stops and outputs a correct estimate.
We can bound this number of calls by $T = \big\lceil \frac{1}{2}\log N \big\rceil$ using the fact that $\|p\| \geq N^{-1/2}$.
The probability that this happens, i.e., that none of the first $\big\lceil \log \|p\|^{-1} \big\rceil$ implementations of the 2-norm tester fails, is at least $1 - \big\lceil \log \|p\|^{-1} \big\rceil \delta' \geq 1-\delta$ if we set $\delta' = \delta/T$.
This proves the success probability of the algorithm.

To bound the runtime, we first note that the $k$-th iteration runs the 2-norm tester with error $\epsilon_k = \epsilon 2^{-k}$ and success probability $1 - \delta'$, which by Theorem \ref{thm:ql2t} requires a number of QW steps
\[
O\left( \frac{2^k \sqrt{t}}{\epsilon} \log \frac{\log N}{\delta} \log^{1/2} \frac{2^k N}{\epsilon} \right).
\]
Now we bound the expected number of iterations.
If the algorithm succeeds, then it aborts after $\big\lceil \log \|p\|^{-1} \big\rceil$ iterations.
If this does not happen, then either it aborts earlier, resulting in a number of iterations smaller than $\big\lceil \log \|p\|^{-1} \big\rceil$, or it aborts later.
However, after $\big\lceil \log \|p\|^{-1} \big\rceil$ iterations, any successful call to the 2-norm tester will abort the algorithm, which happens per iteration with probability at least $1-\delta$.
In such case the expected number of iterations can be bounded by $(1-\delta)^{-1} \leq 2$ under the assumption that $\delta \leq 1/2$.
In any case we see that the expected number of iterations is $O(\log \|p\|^{-1})$.
Now we can use the fact that $\sum_{k=0}^{b} 2^k \log^{1/2} 2^k \in O(2^b \sqrt{b}) \in O(\|p\|^{-1} \log^{1/2} \|p\|)$ for $b \in O(\log \|p\|^{-1})$ to bound the total expected number of QW steps by
\[
O\left( \frac{\sqrt{t}}{\epsilon \|p\|} \log \frac{\log N}{\delta} \log^{1/2} \frac{N}{\epsilon} \right).
\]
This finalizes the proof.
\end{proof}

\end{appendices}

\end{document}